\documentclass{article}%
\usepackage{amsfonts}
\usepackage{amsmath}
\usepackage{fullpage}
\usepackage{amssymb}
\usepackage{graphicx}
\usepackage{hyperref}%
\providecommand{\U}[1]{\protect\rule{.1in}{.1in}}
\newtheorem{theorem}{Theorem}

\newtheorem{corollary}[theorem]{Corollary}

\newenvironment{proof}[1][Proof]{\noindent\textbf{#1.} }{\ \rule{0.5em}{0.5em}}

\numberwithin{equation}{section}
\begin{document}

\title{Sequential, successive, and simultaneous decoders for entanglement-assisted
classical communication}
\author{Shen Chen Xu and Mark M. Wilde\\\textit{School of Computer Science, McGill University,}\\\textit{Montreal, Quebec, H3A 2A7 Canada}}
\maketitle

\begin{abstract}
Bennett \textit{et al}.~showed that allowing shared entanglement between a
sender and receiver before communication begins dramatically simplifies the
theory of quantum channels, and these results suggest that it would be
worthwhile to study other scenarios for entanglement-assisted classical
communication. In this vein, the present paper makes several contributions to
the theory of entanglement-assisted classical communication. First, we
rephrase the Giovannetti-Lloyd-Maccone sequential decoding argument as a more
general \textquotedblleft packing lemma\textquotedblright\ and show that it
gives an alternate way of achieving the entanglement-assisted classical
capacity. Next, we show that a similar sequential decoder can achieve the
Hsieh-Devetak-Winter region for entanglement-assisted classical communication
over a multiple access channel. Third, we prove the existence of a quantum
simultaneous decoder for entanglement-assisted classical communication over a
multiple access channel with two senders. This result implies a solution of
the quantum simultaneous decoding conjecture for unassisted classical
communication over quantum multiple access channels with two senders, but the
three-sender case still remains open (Sen recently and independently solved
this unassisted two-sender case with a different technique). We then leverage
this result to recover the known regions for unassisted and assisted quantum
communication over a quantum multiple access channel, though our proof
exploits a coherent quantum simultaneous decoder. Finally, we determine an
achievable rate region for communication over an entanglement-assisted bosonic
multiple access channel and compare it with the Yen-Shapiro outer bound for
unassisted communication over the same channel.

\end{abstract}

Shared entanglement between a sender and receiver leads to surprises such as
super-dense coding~\cite{PhysRevLett.69.2881}\ and
teleportation~\cite{PhysRevLett.70.1895}, and these protocols were the first
to demonstrate that entanglement, classical bits, and quantum bits can
interact in interesting ways. For this reason, one could argue that these
protocols and their noisy generalizations~\cite{DHW05RI,HW08GFP,HW09}\ make
quantum information theory~\cite{book2000mikeandike,W11} richer than its
classical counterpart~\cite{book1991cover}. A good way to think of the
super-dense coding protocol is that it is a statement of resource
conversion~\cite{DHW05RI}: one noiseless qubit channel and one noiseless ebit
are sufficient to generate two noiseless bit channels between a sender and receiver.

Bennett \textit{et al}.~explored a generalization of the super-dense coding
protocol in which a sender and receiver are given noiseless entanglement in
whatever form they wish and access to many independent uses of a noisy quantum
channel, and the goal is to determine how many asymptotically perfect
noiseless bit channels that the sender and receiver can simulate with the
aforementioned resources~\cite{BSST99,BSST02,H02}. The entanglement-assisted
classical capacity theorem provides a beautiful answer to this question.\ The
optimal rate at which they can communicate classical bits in the presence of
free entanglement is equal to the mutual information of the
channel~\cite{BSST02,H02}, defined as%
\[
I\left(  \mathcal{N}\right)  \equiv\max_{\phi^{AA^{\prime}}}I\left(
A;B\right)  _{\rho},
\]
where $\rho^{AB}\equiv\mathcal{N}^{A^{\prime}\rightarrow B}(\phi^{AA^{\prime}%
})$, $\mathcal{N}^{A^{\prime}\rightarrow B}$ is the noisy channel connecting
the sender to the receiver, and $\phi^{AA^{\prime}}$ is a pure, bipartite
state prepared at the sender's end of the channel. This result is the
strongest statement that quantum information theorists have been able to make
in the theory of quantum channels, because the above channel mutual
information is additive as a function of any two channels $\mathcal{N}$ and
$\mathcal{M}$ ~\cite{PhysRevA.56.3470}:
\[
I\left(  \mathcal{N}\otimes\mathcal{M}\right)  =I\left(  \mathcal{N}\right)
+I\left(  \mathcal{M}\right)  ,
\]
and the mutual information $I\left(  A;B\right)  $ is concave in the input
state when the channel is fixed~\cite{PhysRevA.56.3470} (these two properties
imply that we can actually calculate the entanglement-assisted classical
capacity of \textit{any} quantum channel). Furthermore, this information
measure is particularly robust in the sense that a quantum feedback channel
from receiver to sender does not increase it---Bowen showed that the classical
capacity of a quantum channel in the presence of unlimited quantum feedback
communication is equal to the entanglement-assisted classical
capacity~\cite{B04}. For these reasons, the entanglement-assisted classical
capacity of a quantum channel is the best formal analogy of Shannon's
classical capacity of a classical channel~\cite{bell1948shannon}.

The simplification that shared entanglement brings to the theory of quantum
channels suggests that it might be fruitful to explore other scenarios in
which communicating parties share entanglement, and this is precisely the goal
of the present paper. Indeed, we explore five different scenarios for
entanglement-assisted classical communication:

\begin{enumerate}
\item Sequential decoding for entanglement-assisted classical communication
over a single-sender, single-receiver quantum channel.

\item Sequential and successive decoding for entanglement-assisted classical
communication over a quantum multiple access channel (a two-sender,
single-receiver channel).

\item Simultaneous decoding for classical communication over an
entanglement-assisted quantum multiple access channel.

\item Coherent simultaneous decoding for assisted and unassisted quantum
communication over a quantum multiple access channel.

\item Entanglement-assisted classical communication over a bosonic multiple
access channel.
\end{enumerate}

We briefly overview each of these scenarios in what follows.

Our first contribution is a sequential decoder for entanglement-assisted
classical communication, meaning that the receiver performs a sequence of
measurements with \textquotedblleft yes/no\textquotedblright\ outcomes in
order to determine the message that the sender transmits (the receiver
performs these measurements on the channel outputs and his share of the
entanglement). The idea of this approach is the same as the recent
Giovannetti-Lloyd-Maccone (GLM) sequential decoder for unassisted classical
communication~\cite{GLM10} (which in turn bears similarities to the Feinstein
approach \cite{F54,ON07,itit1999winter}). In fact, our approach for proving
that the sequential method works for the entanglement-assisted case is to
rephrase their argument as a more general \textquotedblleft packing
lemma\textquotedblright~\cite{HDW08,W11}\ and exploit the
entanglement-assisted coding scheme of Hsieh \textit{et al}.~\cite{HDW08,W11}.

Our next contribution is to extend this sequential decoding argument to a
quantum multiple access channel. Winter~\cite{W01}\ and Hsieh \textit{et
al}.~\cite{HDW08}\ have already shown that successive decoding works well for
unassisted and assisted transmission of classical information over a quantum
multiple access channel, respectively. (Here, successive decoding means that
the receiver first decodes one sender's message and follows by decoding the
other sender's message). We show that a receiver can exploit a sequence of
measurements with \textquotedblleft yes/no\textquotedblright\ outcomes to
determine the first sender's message, followed by a different sequence of
\textquotedblleft yes/no\textquotedblright\ measurements to determine the
second sender's message. Thus, our decoder here is both sequential and
successive and generalizes the GLM\ sequential decoding scheme.

Our third contribution is to prove that the receiver of an
entanglement-assisted quantum multiple access channel can exploit a quantum
simultaneous decoder to detect two messages sent by two respective senders. A
simultaneous decoder is different from a successive decoder---it can detect
the two senders' messages asymptotically faithfully as long as their
transmission rates are within the pentagonal rate region of the multiple
access channel \cite{el2010lecture,W01,HDW08}. A simultaneous decoder is more
powerful than a successive decoder for two reasons:

\begin{enumerate}
\item A simultaneous decoder does not require the use of time-sharing in order
to achieve the rate region of the multiple access channel (whereas a
successive decoder requires the use of time-sharing). Thus, the technique
should generalize well to the setting of \textquotedblleft
one-shot\textquotedblright\ information theory~\cite{DR09},\ where
time-sharing does not apply because that theory is concerned with what is
possible with a \textit{single} use of a quantum channel.

\item Nearly every proof in classical network information theory exploits a
simultaneous decoder~\cite{el2010lecture}. Thus, a \textit{quantum}
simultaneous decoder would be of broad interest for a network theory of
quantum information. In particular, the strategy for achieving the best known
achievable rate region of the classical interference channel exploits a
simultaneous decoder~\cite{HK81,el2010lecture}. (An interference channel has
two senders and two receivers, and each sender is interested in communicating
with one particular receiver.)
\end{enumerate}

We should mention that Fawzi \textit{et al}.~could prove the existence of a
quantum simultaneous decoder for certain quantum channels~\cite{FHSSW11}, but
a proof for the general case remained missing and they did not address the
entanglement-assisted case. Though, the results of this paper and recent work
of Sen~\cite{S11a} give a quantum simultaneous decoder for unassisted
communication over a two-sender multiple access channel and solve the
conjecture from Ref.~\cite{FHSSW11} for the two-sender case. It remains
unclear how to prove the conjecture for the case of three senders. The results
of this work might be useful for establishing an achievable rate region for a
quantum interference channel setting in which sender-receiver pairs share
entanglement before communication begins, but this remains the topic of future work.

We then leverage the above result to recover the known regions for assisted
and unassisted quantum communication over a quantum multiple access
channel~\cite{nature2005horodecki,YHD05MQAC,HDW08}. We call the decoder a
\textit{coherent quantum simultaneous decoder} because we construct an
isometry from the above simultaneous decoding POVM, and the isometry is what
enables quantum communication between both senders and the receiver.

Our final contribution is to determine an achievable rate region for
entanglement-assisted classical communication over the multiple access bosonic
channel studied in Ref.~\cite{YS05}. This channel is simply a beamsplitter
with two input ports, where the receiver obtains one output port and the
environment of the channel obtains the other output port. The beamsplitter is
a simplified model for light-based free-space communication in a
multiple-access setting.
In order to calculate the rate region for this setting, we apply the theorem
of Hsieh \textit{et al}.~in Ref.~\cite{HDW08}\ with both senders sharing a
two-mode squeezed vacuum state~\cite{GK04} with the receiver. Since this state
achieves the entanglement-assisted capacity of the single-mode lossy bosonic
channel~\cite{GLMS03,GLMS03a,HW01}, we might suspect that it should do well in
the multiple access setting. Though, it still remains open to determine
whether this strategy is optimal.

\section{Packing Argument for a Sequential Decoder}

\label{sec:sequential-packing}Giovannetti, Lloyd, and Maccone (GLM)\ offered a
scheme for transmitting classical information over a quantum channel that
exploits a sequential decoder~\cite{GLM10}. In their sequential decoding
scheme, the receiver tries to distinguish the transmitted message from a list
of all possible messages one by one until the correct one is identified, by
performing a sequence of projective measurements. We recast this procedure as
a general packing argument in this section, and the next section demonstrates
that the sequential decoding scheme works well for entanglement-assisted
classical communication.

\begin{theorem}
[Sequential Packing]\label{thm:sequential-packing}Let $\{p_{X}\left(
x\right)  ,\rho_{x}\}_{x\in\mathcal{X}}$ be an ensemble of states indexed by
letters in an alphabet $\mathcal{X}$. Each state $\rho_{x}$ has the following
spectral decomposition:%
\begin{equation}
\rho_{x}=\sum_{y}\lambda_{x,y}\left\vert \psi_{x,y}\right\rangle \left\langle
\psi_{x,y}\right\vert ,
\end{equation}
and the expected density operator of the ensemble is as follows:%
\begin{equation}
\rho\equiv\sum_{x\in\mathcal{X}}p_{X}\left(  x\right)  \rho_{x}.
\end{equation}
Suppose there exists a code subspace projector $\Pi$ and codeword subspace
projectors $\left\{  \Pi_{x}\right\}  _{x\in\mathcal{X}}$ such that the
following properties hold for some $D,d\geq0$, $1/2\geq\epsilon>0$, and for
all $x\in\mathcal{X}$:%
\begin{align}
\mathrm{Tr}\left\{  \Pi\rho_{x}\right\}   &  \geq1-\epsilon
,\label{eq:unit-prob-1}\\
\mathrm{Tr}\left\{  \Pi_{x}\rho_{x}\right\}   &  \geq1-\epsilon
,\label{eq:unit-prob-2}\\
\Pi_{x}\rho_{x}\Pi_{x}  &  \geq\frac{1}{d}\Pi_{x},\label{eq:equi-part-1}\\
\Pi\rho\Pi &  \leq\frac{1}{D}\Pi,\label{eq:equi-part-2}\\
\lbrack\Pi_{x},\rho_{x}]  &  =0.
\end{align}
Then corresponding to a message set $\mathcal{M}$, we can construct a random
code $\mathcal{C}=\left\{  c_{m}\right\}  _{m\in\mathcal{M}}$ with $c_{m}%
\in\mathcal{X}$ such that the receiver can reliably distinguish between the
states $\left\{  \rho_{c_{m}}\right\}  _{m\in\mathcal{M}}$ by performing a
sequence of projective measurements using the projectors $\Pi$ and $\Pi_{x}$.
More precisely, suppose that our performance measure is the expectation of the
average success probability where the expectation is with respect to all
possible random choices of codes. Then we can bound this performance measure
from below (as long as $2-\exp\left\{  d\left\vert \mathcal{M}\right\vert
/D\right\}  $ is positive):%
\begin{equation}
\mathbb{E}_{\mathcal{C}}\left\{  \bar{p}_{\text{succ}}\left(  \mathcal{C}%
\right)  \right\}  \geq\left\vert \left(  1-2\epsilon\right)  \left(
2-e^{\frac{d}{D}\left\vert \mathcal{M}\right\vert }\right)  \right\vert ^{2},
\end{equation}
implying that the performance measure becomes arbitrarily close to one if
$D/d$ is large, $\left\vert \mathcal{M}\right\vert \ll D/d$, and $\epsilon$ is
arbitrarily small.
\end{theorem}

\begin{proof}
The proof of this lemma is similar to the GLM proof, and we thus place it in
Appendix~\ref{sec:sequential-packing-proof}.
\end{proof}

\section{Sequential Decoding for Entanglement-Assisted Communication}

In this section, we show an application of the GLM\ sequential decoding scheme
to entanglement-assisted classical communication by exploiting the coding
approach of Hsieh \textit{et al}.~\cite{HDW08}. The approach thus gives
another way of achieving the entanglement-assisted classical capacity of a
quantum channel.

\begin{theorem}
[Entanglement-Assisted Sequential Decoding]\label{thm:ea-sequential}The
sequential decoding scheme can achieve the entanglement-assisted classical
capacity of a quantum channel.
\end{theorem}

\begin{proof}
Suppose that a quantum channel $\mathcal{N}^{A^{\prime}\rightarrow A}$
connects Alice to Bob and that they share many copies of an arbitrary
entangled pure state $\left\vert \phi\right\rangle ^{A^{\prime}A}$:%
\begin{equation}
\left\vert \phi\right\rangle ^{A^{\prime n}A^{n}}\equiv\left(  \left\vert
\phi\right\rangle ^{A^{\prime}A}\right)  ^{\otimes n}=\left\vert
\phi\right\rangle ^{A^{\prime}A}\otimes\left\vert \phi\right\rangle
^{A^{\prime}A}\otimes\dots\otimes\left\vert \phi\right\rangle ^{A^{\prime}A},
\end{equation}
where Alice has access to the system $A^{\prime}$ and Bob has access to the
system $A$. Alice chooses a message from her message set $\mathcal{M}$
uniformly at random, applies a corresponding encoder to her shares $A^{\prime
n}$ of the entanglement, and sends the systems $A^{\prime n}$ to Bob. Later in
the analysis, we would like to be able to \textquotedblleft
pull\textquotedblright\ these encoding operations through the channel so that
they are equivalent to some other operator acting at Bob's end. In order to do
this, we can write the many copies of the shared entanglement as a direct sum
of maximally entangled states~\cite{HDW08,W11}. Starting from the Schmidt
decomposition for one copy of the state $\left\vert \phi\right\rangle $%
\begin{equation}
\left\vert \phi\right\rangle ^{A^{\prime}A}=\sum_{z}\sqrt{p_{Z}\left(
z\right)  }\left\vert z\right\rangle ^{A^{\prime}}\left\vert z\right\rangle
^{A},
\end{equation}
we can derive the following using the method of types~\cite{book1991cover,W11}%
:%
\begin{align}
\left\vert \phi\right\rangle ^{A^{\prime n}A^{n}}  &  =\sum_{z^{n}}%
\sqrt{p_{Z^{n}}\left(  z^{n}\right)  }\left\vert z^{n}\right\rangle
^{A^{\prime n}}\left\vert z^{n}\right\rangle ^{A^{n}}\\
&  =\sum_{t}\sum_{z^{n}\in T_{t}}\sqrt{p_{Z^{n}}\left(  z^{n}\right)
}\left\vert z^{n}\right\rangle ^{A^{\prime n}}\left\vert z^{n}\right\rangle
^{A^{n}}\\
&  =\sum_{t}\sqrt{p_{Z^{n}}\left(  z_{t}^{n}\right)  d_{t}}\frac{1}{d_{t}}%
\sum_{z^{n}\in T_{t}}\left\vert z^{n}\right\rangle ^{A^{\prime n}}\left\vert
z^{n}\right\rangle ^{A^{n}}\\
&  =\sum_{t}\sqrt{p\left(  t\right)  }\left\vert \Phi_{t}\right\rangle
^{A^{\prime n}A^{n}},
\end{align}
where%
\begin{equation}
p\left(  t\right)  \equiv p_{Z^{n}}\left(  z_{t}^{n}\right)  d_{t},
\end{equation}
$T_{t}$ is a type class, $d_{t}$ is the dimension of a type class subspace
$t$, $z_{t}^{n}$ is a representative sequence for the type class $t$, and each
$\left\vert \Phi_{t}\right\rangle ^{A^{\prime n}A^{n}}$ is maximally entangled
on the type class subspace specified by $t$ (see Refs.~\cite{HDW08,W11} for
more details on this approach). Thus, applying an operator acting on type
class subspaces at Alice's end is equivalent to applying the transpose of the
same operator at Bob's end. As in Refs.~\cite{HDW08,W11}, Alice constructs her
encoders using the Heisenberg-Weyl set of operators $\left\{  X\left(
x_{t}\right)  Z\left(  z_{t}\right)  \right\}  _{x_{t},z_{t}}$ that act on
each of the type class subspaces%
\begin{equation}
U\left(  s\right)  \equiv\bigoplus_{t}\left(  -1\right)  ^{b_{t}}X\left(
x_{t}\right)  Z\left(  z_{t}\right)  ,
\end{equation}
where $b_{t}$ determines a phase that is applied to the operators in each
subspace. We denote this unitary by $U\left(  s\right)  $ where $s$ is some
vector that contains all the needed indices $x_{t}$, $z_{t}$ and $b_{t}$. Let
$\mathcal{S}$ denote the set of all such possible vectors. We construct a
random code $\left\{  s_{m}\right\}  _{m\in\mathcal{M}}$ where $s_{m}$ is a
vector chosen uniformly at random from $\mathcal{S}$ and the corresponding set
of encoders is then $\left\{  U\left(  s_{m}\right)  \right\}  _{m\in
\mathcal{M}}$. Since the \textquotedblleft transpose trick\textquotedblright%
\ holds for each of these unitaries, we have that
\begin{equation}
U\left(  s\right)  ^{A^{\prime n}}\left\vert \phi\right\rangle ^{A^{\prime
n}A^{n}}=U^{T}\left(  s\right)  ^{A^{n}}\left\vert \phi\right\rangle
^{A^{\prime n}A^{n}}.
\end{equation}
The induced ensemble at Bob's end is then
\begin{equation}
\left\{  \frac{1}{\left\vert \mathcal{S}\right\vert },\sigma_{s}\right\}
_{s\in\mathcal{S}},
\end{equation}
where%
\begin{align}
\sigma_{s}  &  \equiv U^{T}\left(  s\right)  ^{A^{n}}\rho^{A^{n}B^{n}}U^{\ast
}\left(  s\right)  ^{A^{n}},\\
\rho^{A^{n}B^{n}}  &  \equiv\mathcal{N}^{A^{\prime n}\rightarrow B^{n}}\left(
\left\vert \phi\right\rangle \left\langle \phi\right\vert ^{A^{\prime n}A^{n}%
}\right)  .
\end{align}
Let $\overline{\sigma}$ denote the expected state of the ensemble:%
\begin{equation}
\overline{\sigma}\equiv\frac{1}{\left\vert \mathcal{S}\right\vert }\sum
_{s\in\mathcal{S}}\sigma_{s}.
\end{equation}
We give Bob the following code subspace projector:%
\begin{equation}
\Pi\equiv\Pi_{\delta}^{A^{n}}\otimes\Pi_{\delta}^{B^{n}},
\label{eq:EA-code-projector}%
\end{equation}
and the codeword subspace projectors:%
\begin{equation}
\Pi_{s}\equiv U^{T}\left(  s\right)  ^{A^{n}}\Pi_{\delta}^{A^{n}B^{n}}U^{\ast
}\left(  s\right)  ^{A^{n}}, \label{eq:EA-message-proj}%
\end{equation}
where $\Pi_{\delta}^{A^{n}B^{n}}$, $\Pi_{\delta}^{A^{n}}$, and $\Pi_{\delta
}^{B^{n}}$ are the $\delta$-typical projectors for many copies of the states
$\rho^{A^{n}B^{n}}$, $\rho^{A^{n}}=\mathrm{Tr}_{B}\left\{  \rho^{A^{n}B^{n}%
}\right\}  $ and $\rho^{B^{n}}=\mathrm{Tr}_{A}\left\{  \rho^{A^{n}B^{n}%
}\right\}  $, respectively.

At this point we would like to apply our packing argument from
Theorem~\ref{thm:sequential-packing} and we would like to have the following
conditions hold:%
\begin{align}
\mathrm{Tr}\left\{  \Pi\sigma_{s}\right\}   &  \geq1-\epsilon,\\
\mathrm{Tr}\left\{  \Pi_{s}\sigma_{s}\right\}   &  \geq1-\epsilon,\\
\Pi\overline{\sigma}\Pi &  \leq2^{-n\left(  H\left(  A\right)  _{\rho
}+H\left(  B\right)  _{\rho}-\eta\left(  n,\delta\right)  -\delta\right)  }%
\Pi\\
\Pi_{s}\sigma_{s}\Pi_{s}  &  \geq2^{-n\left(  H\left(  AB\right)  _{\rho
}+\delta\right)  }\Pi_{s},\\
\left[  \Pi_{s},\sigma_{s}\right]   &  =0,
\end{align}
where the function $\eta\left(  n,\delta\right)  $ goes to zero as
$n\rightarrow\infty$ and $\delta\rightarrow0$. The first three conditions are
shown in Refs.~\cite{HDW08,W11}. The fourth condition follows from the
equipartition property of typical subspaces~\cite{W11} and the fact that
$U^{T}U^{\ast}=I$ for any unitary operator $U$. The fifth condition follows
from the fact that the projector $\Pi_{s}$ commutes with the density operator
$\sigma_{s}$. By our packing argument in Theorem~\ref{thm:sequential-packing}
that gives a bound on the expectation of the average success probability,
there exists a particular code, with which Alice can transmit messages from
her set $\mathcal{M}$ and Bob can detect the transmitted state by performing a
series of projective measurements, with its average success probability being
greater than%
\begin{align}
\bar{p}_{\text{succ}}  &  \geq\left\vert \left(  1-2\epsilon\right)  \left(
2-\exp\left\{  2^{-n\left(  H\left(  A\right)  _{\rho}+H\left(  B\right)
_{\rho}-H\left(  AB\right)  _{\rho}-\eta\left(  n,\delta\right)
-2\delta\right)  }\left\vert \mathcal{M}\right\vert \right\}  \right)
\right\vert ^{2}\\
&  =\left\vert \left(  1-2\epsilon\right)  \left(  2-\exp\left\{  2^{-n\left(
I\left(  A\,;\,B\right)  _{\rho}-\eta\left(  n,\delta\right)  -2\delta\right)
}\left\vert \mathcal{M}\right\vert \right\}  \right)  \right\vert ^{2}%
\end{align}
Therefore, Alice can pick the size of $\mathcal{M}$ to be $2^{n\left(
I\left(  A\,;\,B\right)  _{\rho}-\eta\left(  n,\delta\right)  -3\delta\right)
}$, and the rate of communication is then
\begin{equation}
C=\frac{1}{n}\log_{2}\left\vert \mathcal{M}\right\vert =I\left(  A;B\right)
_{\rho}-\eta\left(  n,\delta\right)  -3\delta,
\end{equation}
with the average success probability becoming greater than%
\begin{equation}
\bar{p}_{\text{succ}}\geq\left\vert \left(  1-2\epsilon\right)  \left(
2-\exp\left\{  2^{-n\delta}\right\}  \right)  \right\vert ^{2}.
\end{equation}
Thus, for sufficiently large $n$, the sequential decoding scheme achieves the
entanglement-assisted classical capacity with arbitrarily high success probability.

As a final note, we should clarify a bit further:\ there is a codebook
$\left\{  U\left(  s_{m}\right)  \right\}  _{m\in\mathcal{M}}$ for Alice with
entanglement-assisted quantum codewords of the following form:%
\begin{equation}
U^{A^{\prime n}}\left(  s_{m}\right)  \left\vert \phi\right\rangle ^{A^{\prime
n}A^{n}}.
\end{equation}
If Alice sends message $m$, Bob performs a sequence of measurements in the
following order (assuming a correct sequence of events):%
\begin{equation}
\Pi\rightarrow I-\Pi_{s_{1}}\rightarrow\Pi\rightarrow I-\Pi_{s_{2}}%
\rightarrow\Pi\rightarrow\cdots\rightarrow\Pi\rightarrow\Pi_{s_{m}},
\end{equation}
with $\Pi$ and $\Pi_{s_{i}}$ of the form in (\ref{eq:EA-code-projector}) and
(\ref{eq:EA-message-proj}), respectively.
\end{proof}

\section{Packing Argument for Sequential and Successive Decoding over a
Multiple Access Channel}

We now extend the packing argument from Section~\ref{sec:sequential-packing}%
\ to a multiple-access setting, in which there are two senders and one
receiver. The resulting scheme is both sequential and successive---sequential
in the above sense where the receiver linearly tests one codeword at a time
and successive in the sense that the receiver first decodes one sender's
message and follows by decoding the other sender's message. After doing so, we
then briefly remark how this argument achieves the known strategies for both
unassisted~\cite{W01} and assisted classical communication~\cite{HDW08}.

\begin{theorem}
[Sequential and Successive Decoding]Suppose there exists a doubly-indexed
ensemble of quantum states, where two independent distributions generate the
different indices $x$ and $y$:%
\begin{equation}
\left\{  p_{X}\left(  x\right)  p_{Y}\left(  y\right)  ,\rho_{x,y}\right\}  .
\end{equation}
Averaging with the distributions $p_{X}\left(  x\right)  $ and $p_{Y}\left(
y\right)  $ leads to the following states:%
\begin{equation}
\rho_{x}\equiv\sum_{y}p_{Y}\left(  y\right)  \rho_{x,y}%
,\ \ \ \ \ \ \ \ \ \ \rho_{y}\equiv\sum_{x}p_{X}\left(  x\right)  \rho
_{x,y},\ \ \ \ \ \ \ \ \ \ \rho\equiv\sum_{x,y}p_{X}\left(  x\right)
p_{Y}\left(  y\right)  \rho_{x,y}.
\end{equation}
Suppose that there exist projectors $\Pi_{x}$, $\Pi_{y}$, $\Pi_{x,y}$, and
$\Pi$ such that%
\begin{align}
\mathrm{Tr}\left\{  \Pi\rho_{x}\right\}   &  \geq1-\epsilon,\\
\mathrm{Tr}\left\{  \Pi_{x}\rho_{x}\right\}   &  \geq1-\epsilon,\\
\Pi_{x}\rho_{x}\Pi_{x}  &  \geq\frac{1}{d_{1}^{\left(  -\right)  }}\Pi_{x},\\
\Pi\rho\Pi &  \leq\frac{1}{D_{1}}\Pi,\\
\lbrack\Pi_{x},\rho_{x}]  &  =0.
\end{align}
and%
\begin{align}
\mathrm{Tr}\left\{  \Pi_{x}\rho_{x,y}\right\}   &  \geq1-\epsilon,\\
\mathrm{Tr}\left\{  \Pi_{x,y}\rho_{x,y}\right\}   &  \geq1-\epsilon,\\
\Pi_{x,y}\rho_{x,y}\Pi_{x,y}  &  \geq\frac{1}{d_{2}}\Pi_{x,y},\\
\Pi_{x}\rho_{x}\Pi_{x}  &  \leq\frac{1}{d_{1}^{\left(  +\right)  }}\Pi_{x},\\
\lbrack\Pi_{x,y},\rho_{x,y}]  &  =0.
\end{align}
Suppose that $D_{1}/d_{1}^{\left(  -\right)  }$ is large, $\left\vert
\mathcal{L}\right\vert \ll D_{1}/d_{1}^{\left(  -\right)  }$, $d_{1}^{\left(
+\right)  }/d_{2}$ is large, $\left\vert \mathcal{M}\right\vert \ll
d_{1}^{\left(  +\right)  }/d_{2}$, and $\epsilon$ is arbitrarily small. Then
there exists a sequential and successive decoding scheme for the receiver that
succeeds with high probability, in the sense that the expectation of the
average success probability is arbitrarily high:%
\begin{equation}
\mathbb{E}_{\mathcal{C}}\left\{  \bar{p}_{\text{succ}}\left(  \mathcal{C}%
\right)  \right\}  \geq\left\vert \left(  1-2\epsilon\right)  \left(
2-e^{d_{2}\left\vert \mathcal{M}\right\vert /d_{1}^{\left(  +\right)  }%
}\right)  \right\vert ^{2}-2\sqrt{2\left(  \epsilon+\epsilon^{\prime}\right)
},
\end{equation}
with $\epsilon^{\prime}$ chosen so that%
\begin{equation}
2-e^{d_{1}^{\left(  -\right)  }\left\vert \mathcal{L}\right\vert /D_{1}}%
\geq1-\epsilon^{\prime}.
\end{equation}

\end{theorem}

\begin{proof}
The random construction of the code is similar to that in the proof of
Theorem~\ref{thm:sequential-packing}. Given a message set $\mathcal{L}%
=\left\{  1,2,\dots,\left\vert \mathcal{L}\right\vert \right\}  $, we
construct a code $\mathcal{C}_{1}\equiv\left\{  x\left(  l\right)  \right\}
_{l\in\mathcal{L}}$ for Alice randomly such that each $x\left(  l\right)  $
takes a value $x\in\mathcal{X}$ with probability $p_{X}\left(  x\right)  $.
Similarly, given a message set $\mathcal{M}=\left\{  1,2,\dots,\left\vert
\mathcal{M}\right\vert \right\}  $, we construct a code $\mathcal{C}_{2}%
\equiv\left\{  y\left(  m\right)  \right\}  _{m\in\mathcal{M}}$ for Bob
randomly such that each $y\left(  m\right)  $ takes a value $y\in\mathcal{Y}$
with probability $p_{Y}\left(  y\right)  $. Using this code, Alice chooses a
message $l$ from the message set $\mathcal{L}$, Bob chooses a message $m$ from
the message set $\mathcal{M}$, and they encode their messages in the quantum
codeword $\rho_{x\left(  l\right)  ,y\left(  m\right)  }$.

Suppose that the first sender Alice transmits message $l$ and the second
sender Bob transmits message $m$. Without loss of generality, the receiver
first tries to recover the message that Alice transmits. In order to do so, he
measures $\Pi$ followed by $\Pi_{x\left(  1\right)  }$ to determine if the
transmitted message corresponds to the first codeword of Alice, with
$\Pi_{x\left(  1\right)  }$ corresponding to the outcome YES and $Q_{x\left(
1\right)  }\equiv I-\Pi_{x\left(  1\right)  }$ corresponding to the outcome
NO. Suppose that the outcome is NO. He then measures $\Pi$ to project the
state back into the large subspace. Assuming a correct sequence of events, the
receiver continues and measures $Q_{x\left(  i\right)  }$ and $\Pi$ for
$i\in\left\{  2,\ldots,l-1\right\}  $ until getting to the correct outcome
$\Pi_{x\left(  l\right)  }$. Thus, the sequence of projectors measured is as
follows, under the assumption of a correct sequence of events:%
\begin{equation}
\Pi\rightarrow Q_{x\left(  1\right)  }\rightarrow\Pi\rightarrow Q_{x\left(
2\right)  }\rightarrow\Pi\rightarrow\cdots\rightarrow Q_{x\left(  i\right)
}\rightarrow\Pi\rightarrow\cdots\rightarrow\Pi\rightarrow\Pi_{x\left(
l\right)  }.
\end{equation}
After receiving a YES\ outcome from $\Pi_{x\left(  l\right)  }$, the receiver
assumes that the first sender transmitted message $l$. The receiver then tries
to determine the codeword that Bob transmitted by exploiting the projectors
$\Pi_{x\left(  l\right)  }$ and $\Pi_{x\left(  l\right)  ,y\left(  j\right)
}$. He does this in a similar fashion as above, proceeding in the following
order (again under the assumption of a correct sequence of events):%
\begin{equation}
Q_{x\left(  l\right)  ,y\left(  1\right)  }\rightarrow\Pi_{x\left(  l\right)
}\rightarrow Q_{x\left(  l\right)  ,y\left(  2\right)  }\rightarrow
\Pi_{x\left(  l\right)  }\rightarrow\cdots\rightarrow Q_{x\left(  l\right)
,y\left(  i\right)  }\rightarrow\Pi_{x\left(  l\right)  }\rightarrow
\cdots\rightarrow\Pi_{x\left(  l\right)  }\rightarrow\Pi_{x\left(  l\right)
,y\left(  m\right)  }.
\end{equation}

The POVM\ corresponding to the above measurement strategy is as follows:%
\begin{equation}
\Lambda_{l,m}\equiv M_{l,m}^{\dag}M_{l,m},
\end{equation}
where%
\begin{align}
M_{l,m}  &  \equiv\Pi_{x\left(  l\right)  ,y\left(  m\right)  }\overline
{\overline{Q}}_{x\left(  l\right)  ,y\left(  m-1\right)  }\cdots
\overline{\overline{Q}}_{x\left(  l\right)  ,y\left(  1\right)  }\Pi_{x\left(
l\right)  }\overline{Q}_{x\left(  l-1\right)  }\cdots\overline{Q}_{x\left(
1\right)  },\\
\overline{\Theta}  &  \equiv\Pi\Theta\Pi,\\
\overline{\overline{\Theta}}  &  \equiv\Pi_{x\left(  l\right)  }\Theta
\Pi_{x\left(  l\right)  }.
\end{align}
The average success probability of any particular code $c$ is%
\begin{equation}
\bar{p}_{\text{succ}}\left(  c\right)  \equiv\frac{1}{\left\vert
\mathcal{L}\right\vert \left\vert \mathcal{M}\right\vert }\sum_{l,m}%
\text{Tr}\left\{  \Lambda_{l,m}\rho_{x\left(  l\right)  ,y\left(  m\right)
}\right\}  ,
\end{equation}
and the expectation of the average success probability is%
\begin{multline}
\mathbb{E}_{X,Y}\left\{  \bar{p}_{\text{succ}}\left(  C\right)  \right\}
=\sum_{\substack{x\left(  1\right)  ,\ldots,x\left(  \left\vert \mathcal{L}%
\right\vert \right)  ,\\y\left(  1\right)  ,\ldots,y\left(  \left\vert
\mathcal{M}\right\vert \right)  }}p_{X}\left(  x\left(  1\right)  \right)
\cdots p_{X}\left(  x\left(  \left\vert \mathcal{L}\right\vert \right)
\right)  p_{Y}\left(  y\left(  1\right)  \right)  \cdots\label{eq:succ-prob-1}%
\\
\cdots p_{Y}\left(  y\left(  \left\vert \mathcal{M}\right\vert \right)
\right)  \frac{1}{\left\vert \mathcal{L}\right\vert \left\vert \mathcal{M}%
\right\vert }\sum_{l,m}\text{Tr}\left\{  \Lambda_{l,m}\rho_{x\left(  l\right)
,y\left(  m\right)  }\right\}  .
\end{multline}%
\begin{equation}
=\frac{1}{\left\vert \mathcal{L}\right\vert \left\vert \mathcal{M}\right\vert
}\sum_{l,m}\sum_{x,y}p_{X}\left(  x\right)  p_{Y}\left(  y\right)
\text{Tr}\left\{  \Psi_{x}^{m-1}\left(  \Pi_{x}\Phi^{l-1}\left(  \rho
_{x,y}\right)  \Pi_{x}\right)  \Pi_{x,y}\right\}  ,
\end{equation}
where%
\begin{align}
\Phi\left(  \cdot\right)   &  \equiv\sum_{x}p_{X}\left(  x\right)
\overline{Q}_{x}\left(  \cdot\right)  \overline{Q}_{x},\\
\Psi_{x}\left(  \cdot\right)   &  \equiv\sum_{y}p_{Y}\left(  y\right)
\overline{\overline{Q}}_{x,y}\left(  \cdot\right)  \overline{\overline{Q}%
}_{x,y}.
\end{align}
Observe that we can rewrite the success probability in (\ref{eq:succ-prob-1})
as follows:%
\begin{equation}
\sum_{\substack{x\left(  1\right)  ,\ldots,x\left(  \left\vert \mathcal{L}%
\right\vert \right)  ,\\y\left(  1\right)  ,\ldots,y\left(  \left\vert
\mathcal{M}\right\vert \right)  }}p_{X}\left(  x\left(  1\right)  \right)
\cdots p_{X}\left(  x\left(  \left\vert \mathcal{L}\right\vert \right)
\right)  p_{Y}\left(  y\left(  1\right)  \right)  \cdots p_{Y}\left(  y\left(
\left\vert \mathcal{M}\right\vert \right)  \right)  \frac{1}{\left\vert
\mathcal{L}\right\vert \left\vert \mathcal{M}\right\vert }\sum_{l,m}%
\text{Tr}\left\{  \Gamma_{x,y,l,m}\omega_{x,y,l,m}\right\}  ,
\label{eq:succ-prob-2}%
\end{equation}
where%
\begin{align}
\omega_{x,y,l,m}  &  \equiv\Pi_{x\left(  l\right)  }\overline{Q}_{x\left(
l-1\right)  }\cdots\overline{Q}_{x\left(  1\right)  }\rho_{x\left(  l\right)
,y\left(  m\right)  }\overline{Q}_{x\left(  1\right)  }\cdots\overline
{Q}_{x\left(  l-1\right)  }\Pi_{x\left(  l\right)  },\\
\Gamma_{x,y,l,m}  &  \equiv\overline{\overline{Q}}_{x\left(  l\right)
,y\left(  1\right)  }\cdots\overline{\overline{Q}}_{x\left(  l\right)
,y\left(  m-1\right)  }\Pi_{x\left(  l\right)  ,y\left(  m\right)  }%
\overline{\overline{Q}}_{x\left(  l\right)  ,y\left(  m-1\right)  }%
\cdots\overline{\overline{Q}}_{x\left(  l\right)  ,y\left(  1\right)  }.
\end{align}
We can then obtain the following lower bound on (\ref{eq:succ-prob-2}):%
\begin{multline}
\sum_{\substack{x\left(  1\right)  ,\ldots,x\left(  \left\vert \mathcal{L}%
\right\vert \right)  ,\\y\left(  1\right)  ,\ldots,y\left(  \left\vert
\mathcal{M}\right\vert \right)  }}p_{X}\left(  x\left(  1\right)  \right)
\cdots p_{X}\left(  x\left(  \left\vert \mathcal{L}\right\vert \right)
\right)  p_{Y}\left(  y\left(  1\right)  \right)  \cdots
\label{eq:lower-bound-seq-mac}\\
\cdots p_{Y}\left(  y\left(  \left\vert \mathcal{M}\right\vert \right)
\right)  \frac{1}{\left\vert \mathcal{L}\right\vert \left\vert \mathcal{M}%
\right\vert }\sum_{l,m}\left[  \text{Tr}\left\{  \Gamma_{x,y,l,m}%
\rho_{x\left(  l\right)  ,y\left(  m\right)  }\right\}  -\left\Vert
\rho_{x\left(  l\right)  ,y\left(  m\right)  }-\omega_{x,y,l,m}\right\Vert
_{1}\right]  ,
\end{multline}
by exploiting the following inequality:%
\begin{equation}
\text{Tr}\left\{  \Gamma_{x,y,l,m}\omega_{x,y,l,m}\right\}  \geq
\text{Tr}\left\{  \Gamma_{x,y,l,m}\rho_{x\left(  l\right)  ,y\left(  m\right)
}\right\}  -\left\Vert \rho_{x\left(  l\right)  ,y\left(  m\right)  }%
-\omega_{x,y,l,m}\right\Vert _{1},
\end{equation}
which holds for all positive operators $\Gamma_{x,y,l,m}$, $\omega_{x,y,l,m}$,
and $\rho_{x\left(  l\right)  ,y\left(  m\right)  }$ that have spectrum less
than one. So it remains to show that both Tr$\left\{  \Gamma_{x,y,l,m}%
\rho_{x\left(  l\right)  ,y\left(  m\right)  }\right\}  $ is arbitrarily close
to one and $\left\Vert \rho_{x\left(  l\right)  ,y\left(  m\right)  }%
-\omega_{x,y,l,m}\right\Vert _{1}$ is arbitrarily small when averaging over
all codewords and taking the expectation over random codes. We can apply
Theorem~\ref{thm:sequential-packing} to obtain the following inequality:%
\begin{align}
&  \sum_{\substack{x\left(  1\right)  ,\ldots,x\left(  \left\vert
\mathcal{L}\right\vert \right)  ,\\y\left(  1\right)  ,\ldots,y\left(
\left\vert \mathcal{M}\right\vert \right)  }}p_{X}\left(  x\left(  1\right)
\right)  \cdots p_{X}\left(  x\left(  \left\vert \mathcal{L}\right\vert
\right)  \right)  p_{Y}\left(  y\left(  1\right)  \right)  \cdots p_{Y}\left(
y\left(  \left\vert \mathcal{M}\right\vert \right)  \right)  \text{Tr}\left\{
\omega_{x,y,l,m}\right\} \\
&  \geq\left\vert \left(  1-\epsilon\right)  \left(  2-e^{d_{1}^{\left(
-\right)  }\left\vert \mathcal{L}\right\vert /D_{1}}\right)  \right\vert
^{2}\\
&  \geq\left\vert \left(  1-\epsilon\right)  \left(  1-\epsilon^{\prime
}\right)  \right\vert ^{2}\\
&  \geq1-2\left(  \epsilon+\epsilon^{\prime}\right)  ,
\end{align}
with $\epsilon^{\prime}$ chosen as given in the statement of the theorem. We
can then apply the Gentle Operator Lemma for ensembles (Lemma~9.4.3 in
Ref.~\cite{W11}) to prove the following inequality:%
\begin{equation}
\sum_{\substack{x\left(  1\right)  ,\ldots,x\left(  \left\vert \mathcal{L}%
\right\vert \right)  ,\\y\left(  1\right)  ,\ldots,y\left(  \left\vert
\mathcal{M}\right\vert \right)  }}p_{X}\left(  x\left(  1\right)  \right)
\cdots p_{X}\left(  x\left(  \left\vert \mathcal{L}\right\vert \right)
\right)  p_{Y}\left(  y\left(  1\right)  \right)  \cdots p_{Y}\left(  y\left(
\left\vert \mathcal{M}\right\vert \right)  \right)  \frac{1}{\left\vert
\mathcal{L}\right\vert \left\vert \mathcal{M}\right\vert }\left\Vert
\rho_{x\left(  l\right)  ,y\left(  m\right)  }-\omega_{x,y,l,m}\right\Vert
_{1}\leq2\sqrt{2\left(  \epsilon+\epsilon^{\prime}\right)  }.
\end{equation}
Invoking Theorem~\ref{thm:sequential-packing} one more time gives us the
following lower bound:%
\begin{multline}
\sum_{\substack{x\left(  1\right)  ,\ldots,x\left(  \left\vert \mathcal{L}%
\right\vert \right)  ,\\y\left(  1\right)  ,\ldots,y\left(  \left\vert
\mathcal{M}\right\vert \right)  }}p_{X}\left(  x\left(  1\right)  \right)
\cdots p_{X}\left(  x\left(  \left\vert \mathcal{L}\right\vert \right)
\right)  p_{Y}\left(  y\left(  1\right)  \right)  \cdots p_{Y}\left(  y\left(
\left\vert \mathcal{M}\right\vert \right)  \right)  \frac{1}{\left\vert
\mathcal{L}\right\vert \left\vert \mathcal{M}\right\vert }\sum_{l,m}%
\text{Tr}\left\{  \Gamma_{x,y,l,m}\rho_{x\left(  l\right)  ,y\left(  m\right)
}\right\} \\
\geq\left\vert \left(  1-2\epsilon\right)  \left(  2-e^{d_{2}\left\vert
\mathcal{M}\right\vert /d_{1}^{\left(  +\right)  }}\right)  \right\vert ^{2},
\end{multline}
and this completes the proof of the theorem, by combining the above two
inequalities with the lower bound in (\ref{eq:lower-bound-seq-mac}).
\end{proof}

It is straightforward to apply this packing argument to either unassisted or
assisted transmission of classical information over a quantum multiple access
channel. For the unassisted case, one could exploit Winter's coding scheme
with conditionally typical projectors~\cite{W01}, and we would pick the
parameters as%
\begin{align}
D_{1}  &  =2^{n\left[  H\left(  B\right)  -\delta\right]  },\\
d_{1}^{\left(  +\right)  }  &  =2^{n\left[  H\left(  B|X\right)
-\delta\right]  },\\
d_{1}^{\left(  -\right)  }  &  =2^{n\left[  H\left(  B|X\right)
+\delta\right]  },\\
d_{2}  &  =2^{n\left[  H\left(  B|XY\right)  +\delta\right]  },
\end{align}
so that we would have%
\begin{align}
D_{1}/d_{1}^{\left(  -\right)  }  &  =2^{n\left[  I\left(  X;B\right)
-2\delta\right]  },\\
d_{1}^{\left(  +\right)  }/d_{2}  &  =2^{n\left[  I\left(  Y;B|X\right)
-2\delta\right]  }.
\end{align}
For the entanglement-assisted case, one could exploit the coding structure of
Hsieh \textit{et al}.~\cite{HDW08} that we have discussed throughout this
article, and we would pick the parameters as%
\begin{align}
D_{1}  &  =2^{n\left[  H\left(  A\right)  +H\left(  B\right)  +H\left(
C\right)  -\delta\right]  },\\
d_{1}^{\left(  +\right)  }  &  =2^{n\left[  H\left(  B\right)  +H\left(
AC\right)  -\delta\right]  },\\
d_{1}^{\left(  -\right)  }  &  =2^{n\left[  H\left(  B\right)  +H\left(
AC\right)  +\delta\right]  },\\
d_{2}  &  =2^{n\left[  H\left(  ABC\right)  +\delta\right]  },
\end{align}
so that we would have%
\begin{align}
D_{1}/d_{1}^{\left(  -\right)  }  &  =2^{n\left[  I\left(  A;C\right)
-2\delta\right]  },\\
d_{1}^{\left(  +\right)  }/d_{2}  &  =2^{n\left[  I\left(  B;AC\right)
-2\delta\right]  }.
\end{align}

\section{Entanglement-Assisted Quantum Simultaneous Decoding}

\label{sec:ea-simul}In this section, we prove the existence of a simultaneous
decoder for entanglement-assisted classical communication over a quantum
multiple access channel with two senders. A simultaneous decoder differs from
a successive decoder in the sense that such a decoder allows for the receiver
to reliably detect the messages of both senders with a single measurement as
long as the rates are within the pentagonal rate region specified by Theorem~6
of Ref.~\cite{HDW08} and Theorem~\ref{thm:simul-decoder}\ below (it might also
be helpful to consult Ref.~\cite{el2010lecture}\ to see the difference between
classical successive and simultaneous decoders). The advantage of a
simultaneous decoder over a successive decoder is that there is no need to
invoke time-sharing in order to achieve the Hsieh-Devetak-Winter rate region
of the entanglement-assisted multiple access channel in Ref.~\cite{HDW08}.
Also, an analogous classical decoder is required in order to achieve the
Han-Kobayashi rate region for the classical interference channel~\cite{HK81}
(though it requires a simultaneous decoder for three senders).

Concerning the quantum interference channel, Fawzi \textit{et al}.~made
progress towards demonstrating that a quantized version of the classical
Han-Kobayashi rate region is achievable for classical communication over a
quantum interference channel~\cite{FHSSW11}, though they were only able to
prove this result up to a conjecture regarding the existence of a quantum
simultaneous decoder for general channels. The importance of this conjecture
stems not only from the fact that it would allow for a quantization of the
Han-Kobayashi rate region, but also more broadly from the fact that many
coding theorems in classical network information theory exploit the
simultaneous decoding technique~\cite{el2010lecture}. Thus, having a general
quantum simultaneous decoder for an arbitrary number of senders should allow
for the wholesale import of much of classical network information theory into
quantum network information theory.

Our result below applies only to channels with two senders, and the technique
unfortunately does not generally extend to channels with three senders. Thus,
this important case still remains open as a conjecture. Sen independently
arrived at the results here by exploiting both the proof structure outlined
below and a different technique as well~\cite{S11a}.

\begin{theorem}
[Entanglement-Assisted Simultaneous Decoding]\label{thm:simul-decoder}Suppose
that Alice and Charlie share many copies of an entangled pure state
$\left\vert \phi\right\rangle ^{A^{\prime}A}$ where Alice has access to the
system $A^{\prime}$ and Charlie has access to the system $A$. Similarly, let
Bob and Charlie share many copies of an entangled pure state $\left\vert
\psi\right\rangle ^{B^{\prime}B}$. Let $\mathcal{N}^{A^{\prime}B^{\prime
}\rightarrow C}$ be a multiple access channel that connects Alice and Bob to
Charlie, and let%
\begin{equation}
\rho^{ABC}\equiv\mathcal{N}^{A^{\prime}B^{\prime}\rightarrow C}\left(
\left\vert \phi\right\rangle \left\langle \phi\right\vert ^{A^{\prime}%
A}\otimes\left\vert \psi\right\rangle \left\langle \psi\right\vert
^{B^{\prime}B}\right)  . \label{eq:code-state}%
\end{equation}
Then there exists an entanglement-assisted classical communication code with a
corresponding quantum simultaneous decoder, such that the following rate
region is achievable for $R_{1},R_{2}\geq0$:%
\begin{align}
R_{1}  &  \leq I\left(  A;C|B\right)  _{\rho},\\
R_{2}  &  \leq I\left(  B;C|A\right)  _{\rho},\\
R_{1}+R_{2}  &  \leq I\left(  AB;C\right)  _{\rho},
\end{align}
where the entropies are with respect to the state in (\ref{eq:code-state}).
\end{theorem}

\begin{proof}
Suppose that Alice has a message set $\mathcal{L}$ and Bob has a message set
$\mathcal{M}$ from which they will each choose a message $l\in\mathcal{L}$ and
$m\in\mathcal{M}$ uniformly at random to send to Charlie. They construct
random codes $C_{1}\equiv\left\{  s_{1}\left(  l\right)  \right\}
_{l\in\mathcal{L}}$ and $C_{2}\equiv\left\{  s_{2}\left(  m\right)  \right\}
_{m\in\mathcal{M}}$ in the same way as explained in the proof of
Theorem~\ref{thm:ea-sequential}. Both of them encode their messages by
applying unitary encoders to their respective shares of the entanglement,
giving rise to the following states after applying the transpose trick to each
type class~\cite{HDW08,W11}:%
\begin{align}
\left(  U\left(  s_{1}\left(  l\right)  \right)  ^{A^{\prime n}}\otimes
I^{A^{n}}\right)  \left\vert \phi\right\rangle ^{A^{\prime n}A^{n}}  &
=\left(  I^{A^{\prime n}}\otimes U^{T}\left(  s_{1}\left(  l\right)  \right)
^{A^{n}}\right)  \left\vert \phi\right\rangle ^{A^{\prime n}A^{n}},\\
\left(  U\left(  s_{2}\left(  m\right)  \right)  ^{B^{\prime n}}\otimes
I^{B^{n}}\right)  \left\vert \psi\right\rangle ^{B^{\prime n}B^{n}}  &
=\left(  I^{B^{\prime n}}\otimes U^{T}\left(  s_{2}\left(  m\right)  \right)
^{B^{n}}\right)  \left\vert \psi\right\rangle ^{B^{\prime n}B^{n}}.
\end{align}
Then they both send their share of the state to Charlie over the multiple
access channel~$\mathcal{N}^{A^{\prime}B^{\prime}\rightarrow C}$, giving rise
to a state $\sigma_{l,m}$ at Charlie's receiving end:%
\begin{equation}
\sigma_{l,m}\equiv\left(  U^{T}\left(  s_{1}\left(  l\right)  \right)
^{A^{n}}\otimes U^{T}\left(  s_{2}\left(  m\right)  \right)  ^{B^{n}}\right)
\rho^{A^{n}B^{n}C^{n}}\left(  U^{\ast}\left(  s_{1}\left(  l\right)  \right)
^{A^{n}}\otimes U^{\ast}\left(  s_{2}\left(  m\right)  \right)  ^{B^{n}%
}\right)  .
\end{equation}

Charlie decodes with a simultaneous decoding POVM $\left\{  \Lambda
_{l,m}\right\}  _{l\in\mathcal{L},m\in\mathcal{M}}$, defined as follows:%
\begin{equation}
\Lambda_{l,m}\equiv\left(  \sum_{l^{\prime},m^{\prime}}\Upsilon_{l^{\prime
},m^{\prime}}\right)  ^{-\frac{1}{2}}\Upsilon_{l,m}\left(  \sum_{l^{\prime
},m^{\prime}}\Upsilon_{l^{\prime},m^{\prime}}\right)  ^{-\frac{1}{2}},
\end{equation}
where%
\begin{equation}
\Upsilon_{l,m}\equiv U^{T}\left(  s_{1}\left(  l\right)  \right)  ^{A^{n}}%
\hat{\Pi}_{3}\hat{\Pi}_{2}U^{T}\left(  s_{2}\left(  m\right)  \right)
^{B^{n}}\Pi^{A^{n}B^{n}C^{n}}U^{\ast}\left(  s_{2}\left(  m\right)  \right)
^{B^{n}}\hat{\Pi}_{2}\hat{\Pi}_{3}U^{\ast}\left(  s_{1}\left(  l\right)
\right)  ^{A^{n}},
\end{equation}
and%
\begin{align}
\hat{\Pi}_{1}  &  \equiv\left(  \Pi^{A^{n}}\otimes\Pi^{B^{n}C^{n}}\right)  ,\\
\hat{\Pi}_{2}  &  \equiv\left(  \Pi^{B^{n}}\otimes\Pi^{A^{n}C^{n}}\right)  ,\\
\hat{\Pi}_{3}  &  \equiv\left(  \Pi^{C^{n}}\otimes\Pi^{A^{n}B^{n}}\right)  .
\end{align}
The projectors $\Pi^{A^{n}}$, $\Pi^{B^{n}}$, $\Pi^{C^{n}}$, $\Pi^{A^{n}B^{n}}%
$, $\Pi^{A^{n}C^{n}}$, $\Pi^{B^{n}C^{n}}$ and $\Pi^{A^{n}B^{n}C^{n}}$ are
$\delta$-typical projectors for the state $\rho^{A^{n}B^{n}C^{n}}$ onto the
specified systems after tracing out all other systems.

The average error probability when Alice and Bob choose their messages
independently and uniformly at random is%
\begin{equation}
\bar{p}_{e}\equiv\frac{1}{\left\vert \mathcal{L}\right\vert \cdot\left\vert
\mathcal{M}\right\vert }\sum_{l,m}\mathrm{Tr}\left\{  \left(  I-\Lambda
_{l,m}\right)  \sigma_{l,m}\right\}  .\label{eq:avg-error-crit}%
\end{equation}
We can upper bound this error probability from above\footnote{We are indebted
to Pranab Sen for this observation \cite{S11}~(c.f., versions 1 and 2 of this
paper).} as%
\begin{align}
\bar{p}_{e} &  \leq\frac{1}{\left\vert \mathcal{L}\right\vert \cdot\left\vert
\mathcal{M}\right\vert }\sum_{l,m}\mathrm{Tr}\left\{  \left(  I-\Lambda
_{l,m}\right)  U^{T}\left(  s_{2}\left(  m\right)  \right)  \hat{\Pi}%
_{1}U^{\ast}\left(  s_{2}\left(  m\right)  \right)  \sigma_{l,m}U^{T}\left(
s_{2}\left(  m\right)  \right)  \hat{\Pi}_{1}U^{\ast}\left(  s_{2}\left(
m\right)  \right)  \right\}  +\nonumber\\
&  \qquad\qquad\left\Vert U^{T}\left(  s_{2}\left(  m\right)  \right)
\hat{\Pi}_{1}U^{\ast}\left(  s_{2}\left(  m\right)  \right)  \sigma_{l,m}%
U^{T}\left(  s_{2}\left(  m\right)  \right)  \hat{\Pi}_{1}U^{\ast}\left(
s_{2}\left(  m\right)  \right)  -\sigma_{l,m}\right\Vert _{1}\\
&  \leq\frac{1}{\left\vert \mathcal{L}\right\vert \cdot\left\vert
\mathcal{M}\right\vert }\sum_{l,m}\mathrm{Tr}\left\{  \left(  I-\Lambda
_{l,m}\right)  \theta_{l,m}\right\}  +2\sqrt{\epsilon^{\prime}},
\end{align}
where we define%
\begin{align}
\theta_{l,m} &  \equiv U^{T}\left(  s_{2}\left(  m\right)  \right)  ^{B^{n}%
}\hat{\Pi}_{1}U^{\ast}\left(  s_{2}\left(  m\right)  \right)  ^{B^{n}}%
\sigma_{l,m}U^{T}\left(  s_{2}\left(  m\right)  \right)  ^{B^{n}}\hat{\Pi}%
_{1}U^{\ast}\left(  s_{2}\left(  m\right)  \right)  ^{B^{n}}\\
&  =U^{T}\left(  s_{2}\left(  m\right)  \right)  ^{B^{n}}\hat{\Pi}_{1}%
U^{T}\left(  s_{1}\left(  l\right)  \right)  ^{A^{n}}\rho^{A^{n}B^{n}C^{n}%
}U^{\ast}\left(  s_{1}\left(  l\right)  \right)  ^{A^{n}}\hat{\Pi}_{1}U^{\ast
}\left(  s_{2}\left(  m\right)  \right)  ^{B^{n}}.
\end{align}
The first inequality follows from the inequality%
\begin{equation}
\text{Tr}\left\{  \Gamma\rho\right\}  \leq\text{Tr}\left\{  \Gamma
\sigma\right\}  +\left\Vert \rho-\sigma\right\Vert _{1},
\end{equation}
for any operators $0\leq\Gamma,\rho,\sigma\leq I$ (Corollary~9.1.1 of
Ref.~\cite{W11}). The second inequality follows from the properties of quantum
typicality, the Gentle Operator Lemma (Lemma~9.4.2 of Ref.~\cite{W11}), and
the inequality $\mathrm{Tr}\hat{\{\Pi_{1}}U^{T}\left(  s_{1}\left(  l\right)
\right)  ^{A^{n}}\rho^{A^{n}B^{n}C^{n}}U^{\ast}\left(  s_{1}\left(  l\right)
\right)  ^{A^{n}}\}\geq1-\epsilon^{\prime}$ proved in Ref.~\cite{HDW08}. We
now recall the Hayashi-Nagaoka operator inequality~\cite{hayashi2003general}
which holds for any positive operator $S$ and $T$ such that $0\leq S\leq I$
and $T\geq0$:%
\begin{equation}
I-\left(  S+T\right)  ^{-\frac{1}{2}}S\left(  S+T\right)  ^{-\frac{1}{2}}%
\leq2\left(  I-S\right)  +4T.
\end{equation}
Setting%
\begin{align}
S &  =\Upsilon_{l,m},\\
T &  =\sum_{\left(  l^{\prime},m^{\prime}\right)  \neq\left(  l,m\right)
}\Upsilon_{l^{\prime},m^{\prime}},
\end{align}
and applying the Hayashi-Nagaoka operator inequality, we obtain the following
upper bound on the error probability:%
\begin{equation}
\bar{p}_{e}\leq\frac{1}{\left\vert \mathcal{L}\right\vert \cdot\left\vert
\mathcal{M}\right\vert }\sum_{l,m}\left(  2\mathrm{Tr}\left\{  \left(
I-\Upsilon_{l,m}\right)  \theta_{l,m}\right\}  +4\sum_{\left(  l^{\prime
},m^{\prime}\right)  \neq\left(  l,m\right)  }\mathrm{Tr}\left\{
\Upsilon_{l^{\prime},m^{\prime}}\theta_{l,m}\right\}  \right)  +2\sqrt
{\epsilon}.\label{eq:after-HN}%
\end{equation}
Considering the first term $\mathrm{Tr}\left\{  \left(  I-\Upsilon
_{l,m}\right)  \theta_{l,m}\right\}  $, we can prove that
\begin{equation}
\mathrm{Tr}\left\{  \left(  I-\Upsilon_{l,m}\right)  \theta_{l,m}\right\}
\leq\epsilon^{\prime\prime},
\end{equation}
where $\epsilon^{\prime\prime}$ approaches zero when $n$ becomes large. This
inequality follows from the following inequalities%
\begin{align}
\text{Tr}\left\{  \hat{\Pi}_{1}U^{T}\left(  s_{1}\left(  l\right)  \right)
^{A^{n}}\rho^{A^{n}B^{n}C^{n}}U^{\ast}\left(  s_{1}\left(  l\right)  \right)
^{A^{n}}\right\}   &  \geq1-2\epsilon,\\
\text{Tr}\left\{  \hat{\Pi}_{3}U^{T}\left(  s_{2}\left(  m\right)  \right)
^{B^{n}}\rho^{A^{n}B^{n}C^{n}}U^{\ast}\left(  s_{2}\left(  m\right)  \right)
^{B^{n}}\right\}   &  \geq1-2\epsilon,\\
\text{Tr}\left\{  \hat{\Pi}_{2}U^{T}\left(  s_{2}\left(  m\right)  \right)
^{B^{n}}\rho^{A^{n}B^{n}C^{n}}U^{\ast}\left(  s_{2}\left(  m\right)  \right)
^{B^{n}}\right\}   &  \geq1-2\epsilon,\\
\text{Tr}\left\{  \Pi^{A^{n}B^{n}C^{n}}\rho^{A^{n}B^{n}C^{n}}\right\}   &
\geq1-\epsilon,
\end{align}
(which can be proved with the methods of Ref.~\cite{HDW08}) and by applying
\textquotedblleft measurement on approximately close states\textquotedblright%
\ (Corollary 9.1.1 of Ref.~\cite{W11}) and the Gentle Operator Lemma (Lemma
9.4.2 of Ref.~\cite{W11}) several times.

In order to analyze the second term $\sum_{\left(  l^{\prime},m^{\prime
}\right)  \neq\left(  l,m\right)  }\mathrm{Tr}\left\{  \Upsilon_{l^{\prime
},m^{\prime}}\theta_{l,m}\right\}  $, we need to take the expectation over all
random codes and make several observations about the behavior of the codeword
states under the expectation. Note that the encoding unitaries after the
transpose trick and the channel commute because they act on different systems,
so we can apply the encoding unitaries first. To simplify the calculation, we
first consider only applying a random encoding unitary to the system $A^{n}$:%
\begin{align}
&  \mathbb{E}_{\mathcal{C}_{1}}\left\{  U^{T}\left(  s\right)  ^{A^{n}%
}\left\vert \phi\right\rangle \left\langle \phi\right\vert ^{A^{\prime n}%
A^{n}}U^{\ast}\left(  s\right)  ^{A^{n}}\right\} \nonumber\\
=  &  \frac{1}{\left\vert \mathcal{S}_{1}\right\vert }\sum_{s\in
\mathcal{S}_{1}}U^{T}\left(  s\right)  ^{A^{n}}\left(  \sum_{t}\sqrt{p\left(
t\right)  }\left\vert \Phi_{t}\right\rangle ^{A^{\prime n}A^{n}}\right)
\left(  \sum_{t^{\prime}}\sqrt{p\left(  t^{\prime}\right)  }\left\langle
\Phi_{t^{\prime}}\right\vert ^{A^{\prime n}A^{n}}\right)  U^{\ast}\left(
s\right)  ^{A^{n}}\\
=  &  \sum_{t}p\left(  t\right)  \pi_{t}^{A^{\prime n}}\otimes\pi_{t}^{A^{n}}%
\end{align}
where $\pi_{t}$ is the maximally mixed state on the type subspace $t$. To see
why the last equality holds, we note that when $t=t^{\prime}$, averaging over
all elements in $\mathcal{S}_{1}$ gives rise to the state $\mathrm{Tr}_{A^{n}%
}\left\{  \left\vert \Phi_{t}\right\rangle \left\langle \Phi_{t}\right\vert
^{A^{\prime n}A^{n}}\right\}  \otimes\pi_{t}^{A^{n}}=\pi_{t}^{A^{\prime n}%
}\otimes\pi_{t}^{A^{n}}$; when $t\neq t^{\prime}$, it can be shown that the
whole expression sums up to zero~\cite{HDW08,W11}. Now we can append the other
state at Bob's side and send the overall state through the channel. Therefore,
we have that
\begin{align}
\mathbb{E}_{\mathcal{C}_{1}}\left\{  U^{T}\left(  s\right)  ^{A^{n}}%
\rho^{A^{n}B^{n}C^{n}}U^{\ast}\left(  s\right)  ^{A^{n}}\right\}   &
=\mathcal{N}^{A^{\prime n}B^{\prime n}\rightarrow C^{n}}\left(  \sum
_{t}p\left(  t\right)  \pi_{t}^{A^{\prime n}}\otimes\pi_{t}^{A^{n}}\otimes
\psi^{B^{\prime n}B^{n}}\right) \\
&  =\sum_{t}p\left(  t\right)  \pi_{t}^{A^{n}}\otimes\mathcal{N}^{A^{\prime
n}B^{\prime n}\rightarrow C^{n}}\left(  \pi_{t}^{A^{n}}\otimes\psi^{B^{\prime
n}B^{n}}\right)  .
\end{align}
Now consider the above state sandwiched between the projectors $\hat{\Pi}_{1}%
$:%
\begin{align}
&  \hat{\Pi}_{1}\mathbb{E}_{\mathcal{C}_{1}}\left\{  U^{T}\left(  s\right)
^{A^{n}}\rho^{A^{n}B^{n}C^{n}}U^{\ast}\left(  s\right)  ^{A^{n}}\right\}
\hat{\Pi}_{1}\nonumber\\
&  =\left(  \Pi^{A^{n}}\otimes\Pi^{B^{n}C^{n}}\right)  \left(  \sum
_{t}p\left(  t\right)  \pi_{t}^{A^{n}}\otimes\mathcal{N}^{A^{\prime
n}B^{\prime n}\rightarrow C^{n}}\left(  \pi_{t}^{A^{n}}\otimes\psi^{B^{\prime
n}B^{n}}\right)  \right)  \left(  \Pi^{A^{n}}\otimes\Pi^{B^{n}C^{n}}\right) \\
&  =\sum_{t}p\left(  t\right)  \left(  \Pi^{A^{n}}\pi_{t}^{A^{n}}\Pi^{A^{n}%
}\right)  \otimes\left(  \Pi^{B^{n}C^{n}}\mathcal{N}^{A^{\prime n}B^{\prime
n}\rightarrow C^{n}}\left(  \pi_{t}^{A^{n}}\otimes\psi^{B^{\prime n}B^{n}%
}\right)  \Pi^{B^{n}C^{n}}\right)
\end{align}
At this point, we note that $\pi_{t}^{A^{n}}=\Pi_{t}^{A^{n}}/\mathrm{Tr}%
\left\{  \Pi_{t}^{A^{n}}\right\}  $, $\mathrm{Tr}\left\{  \Pi_{t}^{A^{n}%
}\right\}  \geq2^{n\left(  H\left(  A\right)  -\eta\left(  n,\delta\right)
\right)  }$ for a typical type $t$, and $\Pi^{A^{n}}\Pi_{t}^{A^{n}}\Pi^{A^{n}%
}\leq\Pi^{A^{n}}$, where $\Pi_{t}^{A^{n}}$ is a projector onto the
$t^{\mathrm{th}}$ type class subspace. Therefore, the above expression is
bounded from above by the following one:%
\begin{align}
&  \leq2^{-n\left(  H\left(  A\right)  _{\rho}-\eta\left(  n,\delta\right)
\right)  }\Pi^{A^{n}}\otimes\left(  \Pi^{B^{n}C^{n}}\mathcal{N}^{A^{\prime
n}B^{\prime n}\rightarrow C^{n}}\left(  \sum_{t}p\left(  t\right)  \left(
\pi_{t}^{A^{n}}\right)  \otimes\psi^{B^{\prime n}B^{n}}\right)  \Pi
^{B^{n}C^{n}}\right) \nonumber\\
&  =2^{-n\left(  H\left(  A\right)  _{\rho}-\eta\left(  n,\delta\right)
\right)  }\Pi^{A^{n}}\otimes\left(  \Pi^{B^{n}C^{n}}\mathcal{N}^{A^{\prime
n}B^{\prime n}\rightarrow C^{n}}\left(  \phi^{A^{\prime n}}\otimes
\psi^{B^{\prime n}B^{n}}\right)  \Pi^{B^{n}C^{n}}\right) \\
&  \leq2^{-n\left(  H\left(  A\right)  _{\rho}+H\left(  BC\right)  _{\rho
}-\eta\left(  n,\delta\right)  -c\delta\right)  }\hat{\Pi}_{1}.
\label{eq:pi1sandwich}%
\end{align}
We also note that similar observations can be made when applying random
encoding unitaries to the system $B^{\prime n}$ alone or to both the systems
$A^{\prime n}$ and $B^{\prime n}$.

Now we proceed to bound the second term in the RHS of (\ref{eq:after-HN}) from
above by taking the expectation over the random codes $\mathcal{C}_{1}$ and
$\mathcal{C}_{2}$:%
\begin{multline*}
\mathbb{E}_{\mathcal{C}_{1},\mathcal{C}_{2}}\left\{  \sum_{\left(  l^{\prime
},m^{\prime}\right)  \neq\left(  l,m\right)  }\mathrm{Tr}\left\{
\Upsilon_{l^{\prime},m^{\prime}}\theta_{l,m}\right\}  \right\} \\
=\mathbb{E}_{\mathcal{C}_{1},\mathcal{C}_{2}}\left\{  \sum_{l^{\prime}\neq
l}\mathrm{Tr}\left\{  \Upsilon_{l^{\prime},m}\theta_{l,m}\right\}
+\sum_{m^{\prime}\neq m}\mathrm{Tr}\left\{  \Upsilon_{l,m^{\prime}}%
\theta_{l,m}\right\}  +\sum_{l^{\prime}\neq l,\ m^{\prime}\neq m}%
\mathrm{Tr}\left\{  \Upsilon_{l^{\prime},m^{\prime}}\theta_{l,m}\right\}
\right\}  .
\end{multline*}
We bound the first error on the RHS above, which corresponds to Charlie
correctly identifying the message from Bob only:%
\begin{align}
&  \mathbb{E}_{\mathcal{C}_{1},\mathcal{C}_{2}}\left\{  \sum_{l^{\prime}\neq
l}\mathrm{Tr}\left\{  \Upsilon_{l^{\prime},m}\theta_{l,m}\right\}  \right\}
\nonumber\\
&  =\sum_{l^{\prime}\neq l}\mathbb{E}_{\mathcal{C}_{2}}\left\{  \mathrm{Tr}%
\left\{  \mathbb{E}_{\mathcal{C}_{1}}\left\{  \Upsilon_{l^{\prime},m}\right\}
\mathbb{E}_{\mathcal{C}_{1}}\left\{  \theta_{l,m}\right\}  \right\}  \right\}
\\
&  =\sum_{l^{\prime}\neq l}\mathbb{E}_{\mathcal{C}_{2}}\left\{  \mathrm{Tr}%
\left\{  \mathbb{E}_{\mathcal{C}_{1}}\left\{  \Upsilon_{l^{\prime},m}\right\}
U^{T}\left(  s_{2}\left(  m\right)  \right)  ^{B^{n}}\hat{\Pi}_{1}%
\mathbb{E}_{\mathcal{C}_{1}}\left\{  U^{T}\left(  s_{1}\left(  l\right)
\right)  ^{A^{n}}\rho^{A^{n}B^{n}C^{n}}U^{\ast}\left(  s_{1}\left(  l\right)
\right)  ^{A^{n}}\right\}  \hat{\Pi}_{1}U^{\ast}\left(  s_{2}\left(  m\right)
\right)  ^{B^{n}}\right\}  \right\} \\
&  \leq2^{-n\left(  H\left(  A\right)  _{\rho}+H\left(  BC\right)  _{\rho
}-\eta\left(  n,\delta\right)  -c\delta\right)  }\sum_{l^{\prime}\neq
l}\mathbb{E}_{\mathcal{C}_{2}}\left\{  \mathrm{Tr}\left\{  \mathbb{E}%
_{\mathcal{C}_{1}}\left\{  \Upsilon_{l^{\prime},m}\right\}  U^{T}\left(
s_{2}\left(  m\right)  \right)  ^{B^{n}}\hat{\Pi}_{1}U^{\ast}\left(
s_{2}\left(  m\right)  \right)  ^{B^{n}}\right\}  \right\} \\
&  \leq2^{-n\left(  H\left(  A\right)  _{\rho}+H\left(  BC\right)  _{\rho
}-\eta\left(  n,\delta\right)  -c\delta\right)  }\sum_{l^{\prime}\neq
l}\mathbb{E}_{\mathcal{C}_{1},\mathcal{C}_{2}}\left\{  \mathrm{Tr}\left\{
\Pi^{A^{n}B^{n}C^{n}}\right\}  \right\} \\
&  \leq2^{-n\left(  H\left(  A\right)  _{\rho}+H\left(  BC\right)  _{\rho
}-H\left(  ABC\right)  _{\rho}-\eta\left(  n,\delta\right)  -2c\delta\right)
}\left\vert \mathcal{L}\right\vert \\
&  =2^{-n\left(  I\left(  A;C\mid B\right)  \right)  _{\rho}-\eta\left(
n,\delta\right)  -2c\delta}\left\vert \mathcal{L}\right\vert .
\end{align}
The first equality follows from the fact that the codewords for messages $l$
and $l^{\prime}$ are different and therefore independent (because of the way
that we randomly selected the code). The first inequality follows from our
observation in (\ref{eq:pi1sandwich}).

We now bound the second error term, which corresponds to Charlie correctly
identifying the message from Alice only:%
\begin{align}
&  \mathbb{E}_{\mathcal{C}_{1},\mathcal{C}_{2}}\left\{  \sum_{m^{\prime}\neq
m}\mathrm{Tr}\left\{  \Upsilon_{l,m^{\prime}}\theta_{l,m}\right\}  \right\}
\nonumber\\
&  =\sum_{m^{\prime}\neq m}\mathbb{E}_{\mathcal{C}_{1}}\left\{  \mathrm{Tr}%
\left\{  \mathbb{E}_{\mathcal{C}_{2}}\left\{  \Upsilon_{l,m^{\prime}}\right\}
\mathbb{E}_{\mathcal{C}_{2}}\left\{  \theta_{l,m}\right\}  \right\}  \right\}
\\
&  =\sum_{m^{\prime}\neq m}\mathbb{E}_{\mathcal{C}_{1}}\left\{  \mathrm{Tr}%
\left\{  U^{T}\left(  s_{1}\left(  l\right)  \right)  ^{A^{n}}\hat{\Pi}%
_{3}\hat{\Pi}_{2}\mathbb{E}_{\mathcal{C}_{2}}\left\{  U^{T}\left(
s_{2}\left(  m^{\prime}\right)  \right)  ^{B^{n}}\Pi_{A^{n}B^{n}C^{n}}U^{\ast
}\left(  s_{2}\left(  m^{\prime}\right)  \right)  ^{B^{n}}\right\}  \hat{\Pi
}_{2}\hat{\Pi}_{3}U^{\ast}\left(  s_{1}\left(  l\right)  \right)  ^{A^{n}%
}\mathbb{E}_{\mathcal{C}_{2}}\left\{  \theta_{l,m}\right\}  \right\}
\right\}
\end{align}%
\begin{align}
&  \leq2^{n\left(  H\left(  ABC\right)  _{\rho}+c\delta\right)  }%
\sum_{m^{\prime}\neq m}\mathbb{E}_{\mathcal{C}_{1}}\left\{  \mathrm{Tr}%
\left\{
\begin{array}
[c]{c}%
U^{T}\left(  s_{1}\left(  l\right)  \right)  ^{A^{n}}\hat{\Pi}_{3}\hat{\Pi
}_{2}\mathbb{E}_{\mathcal{C}_{2}}\left\{  U^{T}\left(  s_{2}\left(  m^{\prime
}\right)  \right)  ^{B^{n}}\rho_{A^{n}B^{n}C^{n}}U^{\ast}\left(  s_{2}\left(
m^{\prime}\right)  \right)  ^{B^{n}}\right\}  \cdot\\
\hat{\Pi}_{2}\hat{\Pi}_{3}U^{\ast}\left(  s_{1}\left(  l\right)  \right)
^{A^{n}}\mathbb{E}_{\mathcal{C}_{2}}\left\{  \theta_{l,m}\right\}
\end{array}
\right\}  \right\} \\
&  \leq2^{-n\left(  H\left(  B\right)  _{\rho}+H\left(  AC\right)  _{\rho
}-H\left(  ABC\right)  _{\rho}-\eta\left(  n,\delta\right)  -2c\delta\right)
}\sum_{m^{\prime}\neq m}\mathbb{E}_{\mathcal{C}_{1}}\left\{  \mathrm{Tr}%
\left\{  U^{T}\left(  s_{1}\left(  l\right)  \right)  ^{A^{n}}\hat{\Pi}%
_{3}\hat{\Pi}_{2}\hat{\Pi}_{3}U^{\ast}\left(  s_{1}\left(  l\right)  \right)
^{A^{n}}\mathbb{E}_{\mathcal{C}_{2}}\left\{  \theta_{l,m}\right\}  \right\}
\right\} \\
&  \leq2^{-n\left(  H\left(  B\right)  _{\rho}+H\left(  AC\right)  _{\rho
}-H\left(  ABC\right)  _{\rho}-\eta\left(  n,\delta\right)  -2c\delta\right)
}\sum_{m^{\prime}\neq m}\mathrm{Tr}\left\{  \mathbb{E}_{\mathcal{C}%
_{1},\mathcal{C}_{2}}\left\{  \theta_{l,m}\right\}  \right\} \\
&  =2^{-n\left(  I\left(  B;C|A\right)  _{\rho}-\eta\left(  n,\delta\right)
-2c\delta\right)  }\left\vert \mathcal{M}\right\vert .
\end{align}
The first equality follows because the codewords for messages $m$ and
$m^{\prime}$ are different and thus independent. The first inequality follows
from%
\[
\Pi^{A^{n}B^{n}C^{n}}\leq2^{n[H\left(  ABC\right)  +c\delta]}\Pi^{A^{n}%
B^{n}C^{n}}\rho^{A^{n}B^{n}C^{n}}\Pi^{A^{n}B^{n}C^{n}}\leq2^{n[H\left(
ABC\right)  +c\delta]}\rho^{A^{n}B^{n}C^{n}},
\]
and we applied a similar observation as in (\ref{eq:pi1sandwich}) to obtain
the second inequality.

We now bound the third error term:%
\begin{align}
&  \mathbb{E}_{\mathcal{C}_{1},\mathcal{C}_{2}}\left\{  \sum_{l^{\prime}\neq
l,\ m^{\prime}\neq m}\mathrm{Tr}\left\{  \Upsilon_{l^{\prime},m^{\prime}%
}\theta_{l,m}\right\}  \right\}  \nonumber\\
&  =\sum_{l^{\prime}\neq l,\ m^{\prime}\neq m}\mathrm{Tr}\left\{
\mathbb{E}_{\mathcal{C}_{1},\mathcal{C}_{2}}\left\{  \Upsilon_{l^{\prime
},m^{\prime}}\right\}  \mathbb{E}_{\mathcal{C}_{1},\mathcal{C}_{2}}\left\{
\theta_{l,m}\right\}  \right\}  \\
&  =\sum_{\substack{l^{\prime}\neq l\\m^{\prime}\neq m}}\mathrm{Tr}\left\{
\mathbb{E}_{\mathcal{C}_{1},\mathcal{C}_{2}}\left\{  \Upsilon_{l^{\prime
},m^{\prime}}\right\}  \mathbb{E}_{\mathcal{C}_{2}}\left\{  U^{T}\left(
s_{2}\left(  m\right)  \right)  ^{B^{n}}\hat{\Pi}_{1}\mathbb{E}_{\mathcal{C}%
_{1}}\left\{  U^{T}\left(  s_{1}\left(  l\right)  \right)  ^{A^{n}}\rho
^{A^{n}B^{n}C^{n}}U^{\ast}\left(  s_{1}\left(  l\right)  \right)  ^{A^{n}%
}\right\}  \hat{\Pi}_{1}U^{\ast}\left(  s_{2}\left(  m\right)  \right)
^{B^{n}}\right\}  \right\}  \label{eq:third-term-error}%
\end{align}
Consider the following operator inequalities:%
\begin{multline}
\hat{\Pi}_{1}\mathbb{E}_{\mathcal{C}_{1}}\left\{  U^{T}\left(  s_{1}\left(
l\right)  \right)  ^{A^{n}}\rho^{A^{n}B^{n}C^{n}}U^{\ast}\left(  s_{1}\left(
l\right)  \right)  ^{A^{n}}\right\}  \hat{\Pi}_{1}\\
\leq2^{-n\left(  H\left(  A\right)  _{\rho}-\eta\left(  n,\delta\right)
\right)  }\Pi^{A^{n}}\otimes\Pi^{B^{n}C^{n}}\mathcal{N}\left(  \phi^{A^{\prime
n}}\otimes\psi^{B^{\prime n}B^{n}}\right)  \Pi^{B^{n}C^{n}}\\
\leq2^{-n\left(  H\left(  A\right)  _{\rho}-\eta\left(  n,\delta\right)
\right)  }\Pi^{A^{n}}\otimes\mathcal{N}\left(  \phi^{A^{\prime n}}\otimes
\psi^{B^{\prime n}B^{n}}\right)
\end{multline}
The second inequality follows from the fact that the typical projector
commutes with the state $\rho$ and therefore $\Pi\rho\Pi=\sqrt{\rho}\Pi
\sqrt{\rho}\leq\rho$. Thus the quantity in (\ref{eq:third-term-error}) is
upper bounded by the following one:%
\begin{align}
&  \leq2^{-n\left(  H\left(  A\right)  _{\rho}-\eta\left(  n,\delta\right)
\right)  }\sum_{l^{\prime}\neq l,\ m^{\prime}\neq m}\mathrm{Tr}\left\{
\mathbb{E}_{\mathcal{C}_{1},\mathcal{C}_{2}}\left\{  \Upsilon_{l^{\prime
},m^{\prime}}\right\}  \Pi^{A^{n}}\otimes\mathbb{E}_{\mathcal{C}_{2}}\left\{
U^{T}\left(  s_{2}\left(  m\right)  \right)  ^{B^{n}}\mathcal{N}\left(
\phi^{A^{\prime n}}\otimes\psi^{B^{\prime n}B^{n}}\right)  U^{\ast}\left(
s_{2}\left(  m\right)  \right)  ^{B^{n}}\right\}  \right\}  \\
&  =2^{-n\left(  H\left(  A\right)  _{\rho}-\eta\left(  n,\delta\right)
\right)  }\sum_{l^{\prime}\neq l,\ m^{\prime}\neq m}\mathrm{Tr}\left\{
\mathbb{E}_{\mathcal{C}_{1},\mathcal{C}_{2}}\left\{  \Upsilon_{l^{\prime
},m^{\prime}}\right\}  \Pi^{A^{n}}\otimes\left(  \sum_{t}\pi_{t}^{B^{n}%
}\otimes\mathcal{N}\left(  \phi^{A^{\prime n}}\otimes\pi_{t}^{B^{\prime n}%
}\right)  \right)  \right\}  \\
&  =2^{-n\left(  H\left(  A\right)  _{\rho}-\eta\left(  n,\delta\right)
\right)  }\sum_{\substack{l^{\prime}\neq l\\m^{\prime}\neq m}}\mathrm{Tr}%
\left\{  \mathbb{E}_{\mathcal{C}_{1},\mathcal{C}_{2}}\left\{  \Upsilon
_{l^{\prime},m^{\prime}}\right\}  \hat{\Pi}_{3}\left(  \Pi^{A^{n}}%
\otimes\left(  \sum_{t}\pi_{t}^{B^{n}}\otimes\mathcal{N}\left(  \phi
^{A^{\prime n}}\otimes\pi_{t}^{B^{\prime n}}\right)  \right)  \right)
\hat{\Pi}_{3}\right\}
\end{align}%
\begin{align}
&  \leq2^{-n\left(  H\left(  A\right)  _{\rho}+H\left(  B\right)  _{\rho
}-2\eta\left(  n,\delta\right)  \right)  }\sum_{\substack{l^{\prime}\neq
l\\m^{\prime}\neq m}}\mathrm{Tr}\left\{  \mathbb{E}_{\mathcal{C}%
_{1},\mathcal{C}_{2}}\left\{  \Upsilon_{l^{\prime},m^{\prime}}\right\}
\Pi^{A^{n}}\otimes\Pi^{B^{n}}\otimes\Pi^{C^{n}}\mathcal{N}\left(
\phi^{A^{\prime n}}\otimes\sum_{t}\pi_{t}^{B^{\prime n}}\right)  \Pi^{C^{n}%
}\right\}  \\
&  \leq2^{-n\left(  H\left(  A\right)  _{\rho}+H\left(  B\right)  _{\rho
}+H\left(  C\right)  _{\rho}-2\eta\left(  n,\delta\right)  -c\delta\right)
}\sum_{l^{\prime}\neq l,\ m^{\prime}\neq m}\mathrm{Tr}\left\{  \mathbb{E}%
_{\mathcal{C}_{1},\mathcal{C}_{2}}\left\{  \Upsilon_{l^{\prime},m^{\prime}%
}\right\}  \right\}  \\
&  \leq2^{-n\left(  H\left(  A\right)  _{\rho}+H\left(  B\right)  _{\rho
}+H\left(  C\right)  _{\rho}-2\eta\left(  n,\delta\right)  -c\delta\right)
}\sum_{l^{\prime}\neq l,\ m^{\prime}\neq m}\mathrm{Tr}\left\{  \Pi^{A^{n}%
B^{n}C^{n}}\right\}  \\
&  \leq2^{-n\left(  H\left(  A\right)  _{\rho}+H\left(  B\right)  _{\rho
}+H\left(  C\right)  _{\rho}-H\left(  ABC\right)  _{\rho}-2\eta\left(
n,\delta\right)  -2c\delta\right)  }\left\vert \mathcal{L}\right\vert
\cdot\left\vert \mathcal{M}\right\vert \\
&  =2^{-n\left(  I\left(  AB;C\right)  _{\rho}-2\eta\left(  n,\delta\right)
-2c\delta\right)  }\left\vert \mathcal{L}\right\vert \cdot\left\vert
\mathcal{M}\right\vert
\end{align}
Thus, as long as we choose the message set sizes such that the corresponding
rates obey the inequalities in the statement of the theorem, then this ensures
the existence of a code with vanishing average error probability in the
asymptotic limit of large blocklength $n$.
\end{proof}

\subsection{From Average to Maximal Error}

\label{sec:avg-to-max}The above scheme for entanglement-assisted classical
communication satisfies an average error criterion (as specified in
(\ref{eq:avg-error-crit})), but we would like it to satisfy a stronger maximal
error criterion, where we can guarantee that every message pair has a low
error probability. In the single-sender single-receiver case, the standard
argument is just to invoke Markov's inequality to demonstrate that throwing
away half of the codewords ensures that the error for all codewords is less
than $2\epsilon$ if the original average error probability is less than
$\epsilon~$\cite{book1991cover}. This expurgation then only has a negligible
impact on the rate of the code. We cannot employ such an argument for the
multiple access case because the expurgation does not guarantee that the
resulting expurgated codebook of message pairs decomposes as a product of two
expurgated codebooks. Thus, the argument for average-to-maximal error needs to
be a bit more clever.

Yard \textit{et al}.~introduced a straightforward scheme for constructing a
code with low maximal error from one with low average error~\cite{YHD05MQAC},
based on some ideas in Ref.~\cite{CK11} and some further ideas of their own.
The first idea from Ref.~\cite{CK11} is to suppose that the senders and
receiver have access to uniform common randomness. That is, Alice and Charlie
share some common randomness and so do Bob and Charlie. Let $S$ denote the
Alice-Charlie common randomness and let $T$ denote the Bob-Charlie common
randomness. Based on this common randomness, Alice and Bob each compute $l+S$
and $m+T$, where $l$ is Alice's message and $m$ is Bob's message and the
addition is understood to be modulo the size of the respective message sets.
Alice and Bob then encode according to $l+S$ and $m+T$ and Charlie decodes
these messages. Using his share of the common randomness, he subtracts off $S$
and $T$ to obtain the intended messages $l$ and $m$. Now, the expected error
probability for when Alice and Bob transmit the message pair $\left(
l,m\right)  $, where the expectation is with respect to the common randomness,
becomes as follows:%
\begin{align}
\mathbb{E}_{S,T}\left\{  \mathrm{Tr}\left\{  \left(  I-\Lambda_{l+S,m+T}%
\right)  \sigma_{l+S,m+T}\right\}  \right\}    & =\frac{1}{LM}\sum
_{s,t}\mathrm{Tr}\left\{  \left(  I-\Lambda_{l+s,m+t}\right)  \sigma
_{l+s,m+T}\right\}  \\
& =\frac{1}{LM}\sum_{l,m}\mathrm{Tr}\left\{  \left(  I-\Lambda_{l,m}\right)
\sigma_{l,m}\right\}  .
\end{align}
Thus, it becomes clear that the maximal error criterion for each message pair
$\left(  l,m\right)  $ is equivalent to the average error criterion if the
senders and receiver have access to common randomness.

Yard \textit{et al}.~then take this argument further to show that preshared
common randomness is not actually necessary. The main idea is to divide the
overall number of channel uses into $N+1$ blocks each of length $n$. For the
first round, Alice and Bob use the channel $n$ times to establish common
randomness of respective sizes $2^{nR_{1}}$ and $2^{nR_{2}}$ with Charlie.
Since the common randomness is uniformly distributed and our protocol works
well for the uniform distribution, this round fails with probability no larger
than $\epsilon$. Alice, Bob, and Charlie then use this established common
randomness and the randomized protocol given above for the next $N$ rounds
(the key point is that they can use the same common randomness from the first
round for all of the subsequent $N$ rounds). By choosing $N=1/\sqrt{\epsilon}%
$, the first round establishes common randomness at the negligible rates%
\begin{equation}
\frac{1}{nN}\log2^{nR_{i}}=\sqrt{\epsilon}R_{i},
\end{equation}
while ensuring that the subsequent rounds have an error probability no larger
than $N\epsilon=\sqrt{\epsilon}$. Now, the actual distribution resulting from
the first round is $\epsilon$-close to perfect common randomness, but this
only results in an error probability of $2\sqrt{\epsilon}$ for the $N$-blocked
protocol. The resulting achievable rates for classical communication become
$\left(  1-\sqrt{\epsilon}\right)  R_{i}$ for $i\in\left\{  1,2\right\}  $.
Thus, this blocked scheme shows how to convert a protocol with low average
error probability to one with low maximal error probability.

\subsection{Unassisted Simultaneous Decoding}

A simple corollary of Theorem~\ref{thm:simul-decoder} is a simultaneous
decoder for unassisted classical communication. The proof is virtually
identical to the above proof, but it takes advantage of the proof technique in
Section~III-B of Ref.~\cite{HDW08} (thus we omit the details of the proof).
The below result implies a complete solution of the \textquotedblleft strong
interference\textquotedblright\ case for transmitting classical data over a
quantum interference channel \cite{el2010lecture}, if the encoders are
restricted to product-state inputs. Sen independently obtained a proof of the
below corollary with a different technique~\cite{S11a}.

\begin{corollary}
[Unassisted Simultaneous Decoding]\label{thm:unassisted-simul-decoder}Let
$\mathcal{N}^{A^{\prime}B^{\prime}\rightarrow C}$ be a multiple access channel
that connects Alice and Bob to Charlie, and let%
\begin{equation}
\rho^{XYC}\equiv\sum_{x,y}p_{X}\left(  x\right)  p_{Y}\left(  y\right)
\left\vert x\right\rangle \left\langle x\right\vert ^{X}\otimes\left\vert
y\right\rangle \left\langle y\right\vert ^{Y}\otimes\mathcal{N}^{A^{\prime
}B^{\prime}\rightarrow C}\left(  \rho_{x}^{A^{\prime}}\otimes\sigma
_{y}^{B^{\prime}}\right)  . \label{eq-unassisted-code-state}%
\end{equation}
Then there exists a classical communication code with a corresponding quantum
simultaneous decoder, such that the following rate region is achievable for
$R_{1},R_{2}\geq0$:%
\begin{align}
R_{1}  &  \leq I\left(  X;C|Y\right)  _{\rho},\\
R_{2}  &  \leq I\left(  Y;C|X\right)  _{\rho},\\
R_{1}+R_{2}  &  \leq I\left(  XY;C\right)  _{\rho},
\end{align}
where the entropies are with respect to the state in
(\ref{eq-unassisted-code-state}).
\end{corollary}

\section{Quantum communication over a quantum multiple access channel}

In this section, we recover the previously known achievable rate regions for
assisted and unassisted quantum communication over a multiple access channel
\cite{HDW08,nature2005horodecki,YHD05MQAC}. We do so by employing a coherent
version of the protocol from Section~\ref{sec:ea-simul} (many researchers have
often employed this approach in quantum Shannon
theory~\cite{Har03,DHW05RI,HDW08,W11}). Different from prior work, we show
that this region can be achieved without the need for time sharing---the
simultaneous nature of our decoding scheme guarantees this. Our scheme below
achieves quantum communication by employing the blocked protocol from
Section~\ref{sec:avg-to-max}\ (this is again related to the average versus
maximal error issue).

We first recall the resource inequality formalism of Devetak \textit{et
al}.~\cite{DHW05RI}. We denote one noiseless classical bit channel from Alice
to Bob as $\left[  c\rightarrow c\right]  _{AB}$, one noiseless qubit channel
from Alice to Bob as $\left[  q\rightarrow q\right]  _{AB}$, and one ebit of
entanglement shared between Alice and Bob as $\left[  qq\right]  _{AB}$. We
will also be using a coherent channel from Alice to Bob, which is defined to
implement the map \cite{Har03}:%
\begin{equation}
\left\vert i\right\rangle ^{A}\rightarrow\left\vert i\right\rangle
^{A}\left\vert i\right\rangle ^{B}.
\end{equation}
We denote this communication resource as $\left[  q\rightarrow qq\right]
_{AB}$. Note that a coherent channel is a stronger resource than a classical
channel because it can simulate a classical channel if Alice only sends
computational basis states through it. Furthermore, let $\left\langle
\mathcal{N}\right\rangle $ denote one use of a multiple access channel
$\mathcal{N}^{A^{\prime}B^{\prime}\rightarrow C}$. Resource inequalities
describe ways of consuming some communication resources in order to create
others. For example, the protocol described in Section 4 implements the
following resource inequality:%
\begin{equation}
\left\langle \mathcal{N}\right\rangle +H\left(  A\right)  _{\rho}\left[
qq\right]  _{AC}+H\left(  B\right)  _{\rho}\left[  qq\right]  _{BC}\geq
R_{1}\left[  c\rightarrow c\right]  _{AC}+R_{2}\left[  c\rightarrow c\right]
_{BC}.
\end{equation}

We upgrade our protocol for entanglement-assisted classical communication from
the previous section to one for entanglement-assisted coherent communication
with two senders and one receiver.

\begin{theorem}
The following resource inequality corresponds to an achievable coherent
simultaneous decoding protocol for entanglement-assisted coherent
communication over a noisy multiple access quantum channel$~\mathcal{N}$:%
\begin{equation}
\left\langle \mathcal{N}\right\rangle +H\left(  A\right)  _{\rho}\left[
qq\right]  _{AC}+H\left(  B\right)  _{\rho}\left[  qq\right]  _{BC}\geq
R_{1}\left[  q\rightarrow qq\right]  _{AC}+R_{2}\left[  q\rightarrow
qq\right]  _{BC},
\end{equation}
where%
\begin{equation}
\rho^{ABC}\equiv\mathcal{N}^{A^{\prime}B^{\prime}\rightarrow C}\left(
\phi^{A^{\prime}A}\otimes\psi^{B^{\prime}B}\right)  ,\label{eq:sec5rho}%
\end{equation}
as long as%
\begin{align}
R_{1} &  \leq I\left(  A;C|B\right)  _{\rho},\label{eq:ea-rate-region-cond1}\\
R_{2} &  \leq I\left(  B;C|A\right)  _{\rho},\label{eq:ea-rate-region-cond2}\\
R_{1}+R_{2} &  \leq I\left(  AB;C\right)  _{\rho}%
.\label{eq:ea-rate-region-cond3}%
\end{align}
The entropies are with respect to the state in (\ref{eq:sec5rho}).
\end{theorem}

\begin{proof}
We again exploit the blocked protocol from Section~\ref{sec:avg-to-max}. We
assume that Alice, Bob, and Charlie have already established their common
randomness, and we describe how the protocol operates for the first of the $N$
rounds (the round after the one that establishes common randomness). Let $S$
denote the Alice-Charlie common randomness, and let $T$ denote the Bob-Charlie
common randomness.

Suppose that Alice shares a state with a reference system $R_{A}$:
\begin{equation}
\sum_{j,l=1}^{L}\alpha_{j,l}\left\vert j\right\rangle ^{R_{A}}\left\vert
l\right\rangle ^{A_{1}},
\end{equation}
where $\left\{  \left\vert j\right\rangle \right\}  $ and $\left\{  \left\vert
l\right\rangle \right\}  $ are some orthonormal bases for $R_{A}$ and $A_{1}$
respectively. Similarly, Bob also shares a state with a reference system
$R_{B}$:%
\begin{equation}
\sum_{k,m=1}^{M}\beta_{k,m}\left\vert k\right\rangle ^{R_{B}}\left\vert
m\right\rangle ^{B_{1}},
\end{equation}
where $\left\{  \left\vert k\right\rangle \right\}  $ and $\left\{  \left\vert
m\right\rangle \right\}  $ are some orthonormal bases for $R_{B}$ and $B_{1}$
respectively. The parameters $L$ and $M$ for these states are chosen such that%
\begin{align}
R_{1} &  \approx\frac{1}{n}\log L\leq I\left(  A;C|B\right)  _{\rho
},\label{eq:Schmidt-region-1}\\
R_{2} &  \approx\frac{1}{n}\log M\leq I\left(  B;C|A\right)  _{\rho},\\
R_{1}+R_{2} &  \approx\frac{1}{n}\log\left(  LM\right)  \leq I\left(
AB;C\right)  _{\rho}.\label{eq:Schmidt-region-3}%
\end{align}
Alice would like to simulate the action of a coherent channel on her system
$A_{1}$ to a system $A_{2}$ for Charlie:%
\begin{equation}
\sum_{j,l}\alpha_{j,l}\left\vert j\right\rangle ^{R_{A}}\left\vert
l\right\rangle ^{A_{1}}\rightarrow\sum_{j,l}\alpha_{j,l}\left\vert
j\right\rangle ^{R_{A}}\left\vert l\right\rangle ^{A_{1}}\left\vert
l\right\rangle ^{A_{2}},
\end{equation}
and Bob would like to do the same. We demand that they simulate these
resources with vanishing error in the limit of many channel uses. As before,
Alice and Charlie share many copies of a pure entangled state $\left\vert
\phi\right\rangle ^{AA^{\prime}}$, Bob and Charlie share many copies of
$\left\vert \psi\right\rangle ^{BB^{\prime}}$, and they all have access to
many uses of a noisy multiple access channel $\mathcal{N}^{A^{\prime}%
B^{\prime}\rightarrow C}$.

They have their encoding unitaries $\{U\left(  s_{1}\left(  l\right)  \right)
^{A^{\prime n}}\}_{l}$ and $\{U\left(  s_{2}\left(  m\right)  \right)
^{B^{\prime n}}\}_{m}$ as described in Section~\ref{sec:ea-simul}, and they
employ them now as the following controlled unitaries that act on the systems
$A_{1}$ and $B_{1}$ and their shares of the entanglement:%
\begin{align}
&  \sum_{l}\left\vert l\right\rangle \left\langle l\right\vert ^{A_{1}}\otimes
U\left(  s_{1}\left(  l\right)  \right)  ^{A^{\prime n}},\\
&  \sum_{m}\left\vert m\right\rangle \left\langle m\right\vert ^{B_{1}}\otimes
U\left(  s_{2}\left(  m\right)  \right)  ^{B^{\prime n}}.
\end{align}
The resulting global state after applying these unitaries and the transpose
trick is as follows:%
\begin{equation}
\left(  \sum_{j,l}\alpha_{j,l}\left\vert j\right\rangle ^{R_{A}}\left\vert
l\right\rangle ^{A_{1}}U^{T}\left(  s_{1}\left(  l\right)  \right)  ^{A^{n}%
}\left\vert \phi\right\rangle ^{A^{n}A^{\prime n}}\right)  \otimes\left(
\sum_{k,m}\beta_{k,m}\left\vert k\right\rangle ^{R_{B}}\left\vert
m\right\rangle ^{B_{1}}U^{T}\left(  s_{2}\left(  m\right)  \right)  ^{B^{n}%
}\left\vert \psi\right\rangle ^{B^{n}B^{\prime n}}\right)  .
\end{equation}
Alice and Bob both then apply the respective unitaries $U\left(  s_{1}\left(
S\right)  \right)  ^{A^{\prime n}}$ and $U\left(  s_{2}\left(  T\right)
\right)  ^{B^{\prime n}}$, conditional on their common randomness shared with
Charlie. The resulting state is%
\begin{equation}
\left(  \sum_{j,l}\alpha_{j,l}\left\vert j\right\rangle ^{R_{A}}\left\vert
l\right\rangle ^{A_{1}}U^{T}\left(  s_{1}\left(  l+S\right)  \right)  ^{A^{n}%
}\left\vert \phi\right\rangle ^{A^{n}A^{\prime n}}\right)  \otimes\left(
\sum_{k,m}\beta_{k,m}\left\vert k\right\rangle ^{R_{B}}\left\vert
m\right\rangle ^{B_{1}}U^{T}\left(  s_{2}\left(  m+T\right)  \right)  ^{B^{n}%
}\left\vert \psi\right\rangle ^{B^{n}B^{\prime n}}\right)  ,
\end{equation}
where the addition $l+S$ and $m+T$ is modulo $L$ and $M$, respectively. They
both then send their shares of the states over the multiple access channel
$\mathcal{N}^{A^{\prime}B^{\prime}\rightarrow C}$, whose isometric extension
is $U_{\mathcal{N}}^{A^{\prime}B^{\prime}\rightarrow CE}$ and acts on
$\left\vert \phi\right\rangle ^{A^{n}A^{\prime n}}\otimes\left\vert
\psi\right\rangle ^{B^{n}B^{\prime n}}$ as follows:%
\begin{equation}
\left\vert \varphi\right\rangle ^{A^{n}B^{n}C^{n}E^{n}}\equiv U_{\mathcal{N}%
}^{A^{\prime n}B^{\prime n}\rightarrow C^{n}E^{n}}\left(  \left\vert
\phi\right\rangle ^{A^{n}A^{\prime n}}\otimes\left\vert \psi\right\rangle
^{B^{n}B^{\prime n}}\right)  .
\end{equation}
After the transmission, the overall state becomes%
\begin{equation}
\sum_{j,k,l,m}\alpha_{j,l}\beta_{k,m}\left\vert j\right\rangle ^{R_{A}%
}\left\vert l\right\rangle ^{A_{1}}\left\vert k\right\rangle ^{R_{B}%
}\left\vert m\right\rangle ^{B_{1}}\left(  U^{T}\left(  s_{1}\left(
l+S\right)  \right)  ^{A^{n}}\otimes U^{T}\left(  s_{2}\left(  m+T\right)
\right)  ^{B^{n}}\right)  \left\vert \varphi\right\rangle ^{A^{n}B^{n}%
C^{n}E^{n}}.\label{eq:state-after-mac}%
\end{equation}
Charlie performs the following coherent measurement constructed from the POVM
$\left\{  \Lambda_{p,q}\right\}  $ of Section~\ref{sec:ea-simul}:%
\begin{equation}
\Upsilon=\sum_{p,q}\left(  \sqrt{\Lambda_{p,q}}\right)  ^{A^{n}B^{n}C^{n}%
}\otimes\left\vert p\right\rangle ^{A_{2}}\otimes\left\vert q\right\rangle
^{B_{2}}.
\end{equation}
Given that the original POVM is good on average, in the sense that%
\begin{equation}
\frac{1}{LM}\sum_{l,m}\mathrm{Tr}\left\{  \Lambda_{l,m}\sigma_{l,m}\right\}
\geq1-\epsilon,\label{eq:good-code-avg}%
\end{equation}
for all $\epsilon>0$ and sufficiently large $n$ (where each $\sigma_{l,m}$ is
an entanglement-assisted quantum codeword as before), this coherent
measurement also has little effect on the received state while coherently
copying the basis states in registers $A_{1}$ and $B_{1}$. That is, the
expected fidelity overlap between the states%
\begin{equation}
\Upsilon^{A^{n}B^{n}C^{n}A_{2}B_{2}}\left\vert \omega\right\rangle ^{F},
\end{equation}
and%
\begin{equation}
\sum_{j,k,l,m}\alpha_{j,l}\beta_{k,m}\left\vert j\right\rangle ^{R_{A}%
}\left\vert l\right\rangle ^{A_{1}}\left\vert k\right\rangle ^{R_{B}%
}\left\vert m\right\rangle ^{B_{1}}\left(  U^{T}\left(  s_{1}\left(
l+S\right)  \right)  ^{A^{n}}\otimes U^{T}\left(  s_{2}\left(  m+T\right)
\right)  ^{B^{n}}\right)  \left\vert \varphi\right\rangle ^{A^{n}B^{n}%
C^{n}E^{n}}\left\vert l+S\right\rangle ^{A_{2}}\left\vert m+T\right\rangle
^{B_{2}}%
\end{equation}
is larger than $1-\epsilon$, where $\left\vert \omega\right\rangle $ denotes
the state in (\ref{eq:state-after-mac}), the system $F$ denotes all the
systems $A_{1}B_{1}A^{n}B^{n}C^{n}E^{n}$, and the expectation is with respect
to the common randomness $S$ and $T$. To see why this is true, consider the
following chain of inequalities:%
\begin{align}
&  \frac{1}{LM}\sum_{s,t}\sum_{j,k,l,m}\alpha_{j,l}^{\ast}\beta_{k,m}^{\ast
}\left\langle j\right\vert ^{R_{A}}\left\langle l\right\vert ^{A_{1}%
}\left\langle k\right\vert ^{R_{B}}\left\langle m\right\vert ^{B_{1}%
}\left\langle \varphi\right\vert \left(  U^{\ast}\left(  s_{1}\left(
l+s\right)  \right)  ^{A^{n}}\otimes U^{\ast}\left(  s_{2}\left(  m+t\right)
\right)  ^{B^{n}}\right)  \left\langle l+s\right\vert ^{A_{2}}\left\langle
m+t\right\vert ^{B_{2}}\nonumber\\
&  \left(  \sum_{p,q}\left(  \sqrt{\Lambda_{p,q}}\right)  ^{A^{n}B^{n}C^{n}%
}\left\vert p\right\rangle ^{A_{2}}\left\vert q\right\rangle ^{B_{2}}\right)
\nonumber\\
&  \sum_{j^{\prime},k^{\prime},l^{\prime},m^{\prime}}\alpha_{j^{\prime
},l^{\prime}}\beta_{k^{\prime},m^{\prime}}\left\vert j^{\prime}\right\rangle
^{R_{A}}\left\vert l^{\prime}\right\rangle ^{A_{1}}\left\vert k^{\prime
}\right\rangle ^{R_{B}}\left\vert m^{\prime}\right\rangle ^{B_{1}}\left(
U^{T}\left(  s_{1}\left(  l^{\prime}+s\right)  \right)  ^{A^{n}}\otimes
U^{T}\left(  s_{2}\left(  m^{\prime}+t\right)  \right)  ^{B^{n}}\right)
\left\vert \varphi\right\rangle \nonumber\\
&  =\frac{1}{LM}\sum_{s,t}\sum_{j,k,l,m}\left\vert \alpha_{j,l}\right\vert
^{2}\left\vert \beta_{k,m}\right\vert ^{2}\left\langle \varphi\right\vert
\left(  U^{\ast}\left(  s_{1}\left(  l+s\right)  \right)  ^{A^{n}}\otimes
U^{\ast}\left(  s_{2}\left(  m+t\right)  \right)  ^{B^{n}}\right)
\sqrt{\Lambda_{l+s,m+t}}\ \ \times\nonumber\\
&  \ \ \ \ \ \ \ \ \ \ \ \ \ \ \ \left(  U^{T}\left(  s_{1}\left(  l+s\right)
\right)  ^{A^{n}}\otimes U^{T}\left(  s_{2}\left(  m+t\right)  \right)
^{B^{n}}\right)  \left\vert \varphi\right\rangle \\
&  =\sum_{j,k,l,m}\left\vert \alpha_{j,l}\right\vert ^{2}\left\vert
\beta_{k,m}\right\vert ^{2}\frac{1}{LM}\sum_{s,t}\left\langle \varphi
\right\vert \left(  U^{\ast}\left(  s_{1}\left(  l+s\right)  \right)  ^{A^{n}%
}\otimes U^{\ast}\left(  s_{2}\left(  m+t\right)  \right)  ^{B^{n}}\right)
\sqrt{\Lambda_{l+s,m+t}}\ \ \times\nonumber\\
&  \ \ \ \ \ \ \ \ \ \ \ \ \ \ \ \left(  U^{T}\left(  s_{1}\left(  l+s\right)
\right)  ^{A^{n}}\otimes U^{T}\left(  s_{2}\left(  m+t\right)  \right)
^{B^{n}}\right)  \left\vert \varphi\right\rangle \\
&  \geq\sum_{j,k,l,m}\left\vert \alpha_{j,l}\right\vert ^{2}\left\vert
\beta_{k,m}\right\vert ^{2}\times\nonumber\\
&  \frac{1}{LM}\sum_{s,t}\text{Tr}\left\{  \left(  U^{T}\left(  s_{1}\left(
l+s\right)  \right)  ^{A^{n}}\otimes U^{T}\left(  s_{2}\left(  m+t\right)
\right)  ^{B^{n}}\right)  \left\vert \varphi\right\rangle \left\langle
\varphi\right\vert \left(  U^{\ast}\left(  s_{1}\left(  l+s\right)  \right)
^{A^{n}}\otimes U^{\ast}\left(  s_{2}\left(  m+t\right)  \right)  ^{B^{n}%
}\right)  \Lambda_{l+s,m+t}\right\}  \\
&  \geq\sum_{j,k,l,m}\left\vert \alpha_{j,l}\right\vert ^{2}\left\vert
\beta_{k,m}\right\vert ^{2}\left(  1-\epsilon\right)  \\
&  =1-\epsilon
\end{align}
where the last inequality follows from (\ref{eq:good-code-avg}). Thus, the
resulting state is $2\sqrt{\epsilon}$-close in expected trace distance to the
following state:%
\begin{equation}
\sum_{j,k,l,m}\alpha_{j,l}\beta_{k,m}\left\vert j\right\rangle ^{R_{A}%
}\left\vert l\right\rangle ^{A_{1}}\left\vert k\right\rangle ^{R_{B}%
}\left\vert m\right\rangle ^{B_{1}}\left(  U^{T}\left(  s_{1}\left(
l+S\right)  \right)  ^{A^{n}}\otimes U^{T}\left(  s_{2}\left(  m+T\right)
\right)  ^{B^{n}}\right)  \left\vert \varphi\right\rangle ^{A^{n}B^{n}%
C^{n}E^{n}}\left\vert l+S\right\rangle ^{A_{2}}\left\vert m+T\right\rangle
^{B_{2}}.
\end{equation}
Now Charlie performs the following controlled unitary:%
\begin{equation}
\sum_{l,m}\left\vert l\right\rangle \left\langle l\right\vert ^{A_{2}}%
\otimes\left\vert m\right\rangle \left\langle m\right\vert ^{B_{2}}%
\otimes\left(  U^{\ast}\left(  s_{1}\left(  l\right)  \right)  ^{A^{n}}\otimes
U^{\ast}\left(  s_{2}\left(  m\right)  \right)  ^{B^{n}}\right)  ,
\end{equation}
and the resulting state is as follows:%
\begin{equation}
\sum_{j,k,l,m}\alpha_{j,l}\beta_{k,m}\left\vert j\right\rangle ^{R_{A}%
}\left\vert l\right\rangle ^{A_{1}}\left\vert k\right\rangle ^{R_{B}%
}\left\vert m\right\rangle ^{B_{1}}\left\vert \varphi\right\rangle
^{A^{n}B^{n}C^{n}E^{n}}\left\vert l+S\right\rangle ^{A_{2}}\left\vert
m+T\right\rangle ^{B_{2}}%
\end{equation}
Charlie then performs the generalized Pauli shifts $X^{A_{2}}\left(
-S\right)  $ and $X^{B_{2}}\left(  -T\right)  $ (based on his common
randomness) to produce the state%
\begin{equation}
\left(  \sum_{j,l}\alpha_{j,l}\left\vert j\right\rangle ^{R_{A}}\left\vert
l\right\rangle ^{A_{1}}\left\vert l\right\rangle ^{A_{2}}\right)
\otimes\left(  \sum_{k,m}\beta_{k,m}\left\vert k\right\rangle ^{R_{B}%
}\left\vert m\right\rangle ^{B_{1}}\left\vert m\right\rangle ^{B_{2}}\right)
\otimes\left\vert \varphi\right\rangle ^{A^{n}B^{n}C^{n}E^{n}},
\end{equation}
so that Alice and Bob have successfully generated coherent channels with the
receiver Charlie for this round.

The above scheme constitutes the first of the $N$ blocks after establishing
the common randomness. Alice, Bob, and Charlie perform the same scheme for the
next $N-1$ blocks, and they use the same common randomness for each round. For
similar reasons as given at the end of Section~\ref{sec:avg-to-max}, this
scheme works well if we set the number $N$ of rounds equal to $\epsilon
^{-1/4}$ (we require $\epsilon^{-1/4}$ this time because each round disturbs
the state by $2\sqrt{\epsilon}$ so that the overall disturbance for all $N$
rounds is no larger than $N\left(  2\sqrt{\epsilon}\right)  =2\epsilon^{1/4}$).
\end{proof}

The coherent communication identity is a helpful tool in quantum Shannon
theory, and it results from the protocols coherent teleportation and coherent
super-dense coding~\cite{Har03,W11}. It states that two coherent channels are
equivalent to a noiseless quantum channel and noiseless entanglement:%
\begin{equation}
2\log d\left[  q\rightarrow qq\right]  =\log d\left[  q\rightarrow q\right]
+\log d\left[  qq\right]  ,
\end{equation}
where $d$ is the dimension of the underlying systems. Employing this identity
gives us the following achievable rate region for entanglement-assisted
quantum communication:

\begin{corollary}
There exists an entanglement-assisted quantum communication protocol with a
coherent quantum simultaneous decoder if the rates $\widetilde{R}_{1}$ and
$\widetilde{R}_{2}$ of quantum communication satisfy the following
inequalities:%
\begin{align}
\widetilde{R}_{1} &  \leq\frac{1}{2}I\left(  A;C|B\right)  _{\rho},\\
\widetilde{R}_{2} &  \leq\frac{1}{2}I\left(  B;C|A\right)  _{\rho},\\
\widetilde{R}_{1}+\widetilde{R}_{2} &  \leq\frac{1}{2}I\left(  AB;C\right)
_{\rho}.
\end{align}

\end{corollary}

\begin{proof}
We simply recall the resource inequality from the previous theorem and apply
the coherent communication identity:
\begin{align}
\left\langle \mathcal{N}\right\rangle +H\left(  A\right)  _{\rho}\left[
qq\right]  _{AC}+H\left(  B\right)  _{\rho}\left[  qq\right]  _{BC} &  \geq
R_{1}\left[  q\rightarrow qq\right]  _{AC}+R_{2}\left[  q\rightarrow
qq\right]  _{BC}\\
&  \geq\frac{1}{2}R_{1}\left[  qq\right]  _{AC}+\frac{1}{2}R_{1}\left[
q\rightarrow q\right]  _{AC}+\frac{1}{2}R_{2}\left[  qq\right]  _{BC}+\frac
{1}{2}R_{2}\left[  q\rightarrow q\right]  _{BC}.
\end{align}
Throughout out the rest of this section, we will assume that $R_{1}$ and
$R_{2}$ satisfy the conditions (\ref{eq:ea-rate-region-cond1}%
)-(\ref{eq:ea-rate-region-cond3}). If we allow catalytic protocols, that is we
allow the use of some resources for free, provided that they are returned at
the end of the protocol, then we obtain a protocol for entanglement-assisted
quantum communication over a multiple access channel that implements the
following resource inequality:
\begin{equation}
\left\langle \mathcal{N}\right\rangle +\left(  H\left(  A\right)  _{\rho
}-\frac{1}{2}R_{1}\right)  \left[  qq\right]  _{AC}+\left(  H\left(  B\right)
_{\rho}-\frac{1}{2}R_{2}\right)  \left[  qq\right]  _{BC}\geq\frac{1}{2}%
R_{1}\left[  q\rightarrow q\right]  _{AC}+\frac{1}{2}R_{2}\left[  q\rightarrow
q\right]  _{BC}.
\end{equation}

\end{proof}

Combining the above protocol further with entanglement distribution $\left[
q\rightarrow q\right]  \geq\left[  qq\right]  $ gives the following corollary:

\begin{corollary}
There exists a catalytic quantum communication protocol (that consumes no net
entanglement) with a coherent quantum simultaneous decoder if the rates
$S_{1}$ and $S_{2}$ of quantum communication satisfy the following
inequalities:%
\begin{align}
S_{1} &  \leq I\left(  A\rangle C|B\right)  _{\rho},\\
S_{2} &  \leq I\left(  B\rangle C|A\right)  _{\rho},\\
S_{1}+S_{2} &  \leq I\left(  AB\rangle C\right)  _{\rho}.
\end{align}

\end{corollary}

\begin{proof}
The protocol from the above corollary in turn leads to a proof of an
achievable rate region for unassisted quantum communication over a multiple
access channel:%
\begin{align}
&  \left\langle \mathcal{N}\right\rangle +\left(  H\left(  A\right)  _{\rho
}-\frac{1}{2}R_{1}\right)  \left[  qq\right]  _{AC}+\left(  H\left(  B\right)
_{\rho}-\frac{1}{2}R_{2}\right)  \left[  qq\right]  _{BC}\nonumber\\
\geq &  \frac{1}{2}R_{1}\left[  q\rightarrow q\right]  _{AC}+\frac{1}{2}%
R_{2}\left[  q\rightarrow q\right]  _{BC}\\
\geq &  \left(  R_{1}-H\left(  A\right)  _{\rho}\right)  \left[  q\rightarrow
q\right]  _{AC}+\left(  R_{2}-H\left(  B\right)  _{\rho}\right)  \left[
q\rightarrow q\right]  _{BC}+\left(  H\left(  A\right)  _{\rho}-\frac{1}%
{2}R_{1}\right)  \left[  qq\right]  _{AC}+\left(  H\left(  B\right)  _{\rho
}-\frac{1}{2}R_{2}\right)  \left[  qq\right]  _{BC}.
\end{align}
The second inequality follows from the fact that we can perform entanglement
distribution using noiseless quantum channels. After resource cancellation,
this leads to
\begin{align}
\left\langle \mathcal{N}\right\rangle  &  \geq\left(  R_{1}-H\left(  A\right)
_{\rho}\right)  \left[  q\rightarrow q\right]  _{AC}+\left(  R_{2}-H\left(
B\right)  _{\rho}\right)  \left[  q\rightarrow q\right]  _{BC}\\
&  =S_{1}\left[  q\rightarrow q\right]  _{AC}+S_{1}\left[  q\rightarrow
q\right]  _{BC},
\end{align}
where%
\begin{align}
S_{1} &  \leq I\left(  A\rangle B|C\right)  _{\rho},\\
S_{2} &  \leq I\left(  B\rangle A|C\right)  _{\rho},\\
S_{1}+S_{2} &  \leq I\left(  AB\rangle C\right)  _{\rho}.
\end{align}
Again, both of these two capacity regions can be achieved without time
sharing, thanks to our simultaneous decoder.
\end{proof}

\section{Entanglement-Assisted Bosonic Multiple Access Channel}

This final section details our last contribution---an achievable rate region
for entanglement-assisted classical communication over a bosonic multiple
access channel (see Refs.~\cite{WPGCRSL11,EW07}\ for a nice review of bosonic
channels). Perhaps the simplest model for this channel is the following
beamsplitter transformation (Yen and Shapiro~\cite{YS05} considered unassisted
communication over such a channel):%
\begin{align}
\hat{c}  &  =\sqrt{\eta}\hat{a}+\sqrt{1-\eta}\hat{b},\\
\hat{e}  &  =-\sqrt{1-\eta}\hat{a}+\sqrt{\eta}\hat{b},
\end{align}
where $\hat{a}$ is the annihilation operator representing the first sender
Alice's input signal, $\hat{b}$ is the annihilation operator representing the
second sender Bob's input signal, $\hat{c}$ is the annihilation operator for
the receiver's output, and $\hat{e}$ is the annihilation operator for an
inaccessible environment output of the channel. We prove the following theorem:

\begin{theorem}
\label{thm:ea-bosonic} Suppose that Alice is allowed a mean photon number
$N_{S_{a}}$ at her transmitter and Bob is allowed a mean photon number
$N_{S_{b}}$ at his transmitter. Then the following rate region is achievable
for entanglement-assisted transmission of classical information over the
beamsplitter quantum multiple access channel:%
\begin{align}
R_{1}  &  \leq g\left(  N_{S_{a}}\right)  +g\left(  \left(  \lambda_{BC}%
^{+}+1\right)  /2-1\right)  +g\left(  \left(  \lambda_{BC}^{-}+1\right)
/2-1\right)  -g\left(  \eta N_{S_{b}}+\left(  1-\eta\right)  N_{S_{a}}\right)
,\label{eq:EA-boson-region-1}\\
R_{2}  &  \leq g\left(  N_{S_{b}}\right)  +g\left(  \left(  \lambda_{AC}%
^{+}+1\right)  /2-1\right)  +g\left(  \left(  \lambda_{AC}^{-}+1\right)
/2-1\right)  -g\left(  \eta N_{S_{b}}+\left(  1-\eta\right)  N_{S_{a}}\right)
,\\
R_{1}+R_{2}  &  \leq g\left(  N_{S_{a}}\right)  +g\left(  N_{S_{b}}\right)
+g\left(  \eta N_{S_{a}}+\left(  1-\eta\right)  N_{S_{b}}\right)  -g\left(
\eta N_{S_{b}}+\left(  1-\eta\right)  N_{S_{a}}\right)  ,
\label{eq:EA-boson-region-3}%
\end{align}
where%
\begin{align}
g\left(  N\right)   &  \equiv\left(  N+1\right)  \log\left(  N+1\right)
-N\log N,\\
\lambda_{AC}^{\left(  \pm\right)  }  &  =\left(  1-\eta\right)  \left\vert
N_{S_{a}}-N_{S_{b}}\right\vert \pm\sqrt{\left(  1-\eta\right)  ^{2}\left(
N_{S_{a}}-N_{S_{b}}\right)  ^{2}+2\left(  1-\eta\right)  \left(  2N_{S_{a}%
}N_{S_{b}}+N_{S_{a}}+N_{S_{b}}\right)  +1},\\
\lambda_{BC}^{\left(  \pm\right)  }  &  =\eta\left\vert N_{S_{a}}-N_{S_{b}%
}\right\vert \pm\sqrt{\eta^{2}\left(  N_{S_{a}}-N_{S_{b}}\right)  ^{2}%
+2\eta\left(  2N_{S_{a}}N_{S_{b}}+N_{S_{a}}+N_{S_{b}}\right)  +1}.
\end{align}
(Observe that $\lambda_{AC}^{\left(  \pm\right)  }$ and $\lambda_{BC}^{\left(
\pm\right)  }$ are related by the substitution $\eta\leftrightarrow1-\eta$.)
\end{theorem}

\begin{proof}
We assume the most natural entangled states that Alice and Charlie and Bob and
Charlie can share: a two-mode squeezed vacuum~\cite{GK04,WPGCRSL11}. This
state has the following form:%
\begin{equation}
\sum_{n=0}^{\infty}\sqrt{\frac{N_{S}^{n}}{\left(  N_{S}+1\right)  ^{n+1}}%
}\left\vert n\right\rangle \left\vert n\right\rangle ,
\end{equation}
where $N_{S}$ is the average number of photons in one mode (after tracing over
the other), Alice or Bob has the first mode, and Charlie has the second mode.
The covariance matrix for such a state is as follows~\cite{WPGCRSL11}:%
\begin{equation}
V_{\text{TMS}}\left(  N_{S}\right)  \equiv%
\begin{bmatrix}
2N_{S}+1 & 0 & 2\sqrt{N_{S}\left(  N_{S}+1\right)  } & 0\\
0 & 2N_{S}+1 & 0 & -2\sqrt{N_{S}\left(  N_{S}+1\right)  }\\
2\sqrt{N_{S}\left(  N_{S}+1\right)  } & 0 & 2N_{S}+1 & 0\\
0 & -2\sqrt{N_{S}\left(  N_{S}+1\right)  } & 0 & 2N_{S}+1
\end{bmatrix}
.
\end{equation}
The covariance matrix for the overall state before the channel acts is as
follows:%
\begin{equation}
V^{AA^{\prime}BB^{\prime}}\equiv V_{\text{TMS}}\left(  N_{S_{a}}\right)
\oplus V_{\text{TMS}}\left(  N_{S_{b}}\right)  ,
\end{equation}
where $N_{S_{a}}$ is the average number of photons in one share of the state
that Alice shares with Charlie and $N_{S_{b}}$ is the average number of
photons in one share of the state that Bob shares with Charlie.

The symplectic operator for a beamsplitter unitary is as
follows~\cite{WPGCRSL11}:%
\begin{equation}
S_{\text{BS}}^{A^{\prime}B^{\prime}}\equiv%
\begin{bmatrix}
\sqrt{\eta}I & \sqrt{1-\eta}I\\
-\sqrt{1-\eta}I & \sqrt{\eta}I
\end{bmatrix}
,
\end{equation}
and the covariance matrix of the state resulting from the beamsplitter
interaction is%
\begin{equation}
V^{ACBE}\equiv\left(  S_{\text{BS}}^{A^{\prime}B^{\prime}}\oplus
I^{AB}\right)  V^{AA^{\prime}BB^{\prime}}\left(  (S_{\text{BS}}^{A^{\prime
}B^{\prime}})^{T}\oplus I^{AB}\right)  ,
\end{equation}
where modes $C$ and $E$ emerge from the output ports of the beamsplitter (with
input ports $A^{\prime}$ and $B^{\prime}$).

Hsieh \textit{et al}.~proved that the following rate region is achievable for
entanglement-assisted communication over a quantum multiple access channel
$\mathcal{M}$:%
\begin{align}
R_{1}  &  \leq I\left(  A;BC\right)  _{\rho},\\
R_{2}  &  \leq I\left(  B;AC\right)  _{\rho},\\
R_{1}+R_{2}  &  \leq I\left(  AB;C\right)  _{\rho},
\end{align}
where $\rho^{ABC}$ is a state of the following form:%
\begin{equation}
\rho^{ABC}\equiv\mathcal{M}^{A^{\prime}B^{\prime}\rightarrow C}(\phi
^{AA^{\prime}}\otimes\psi^{BB^{\prime}}),
\end{equation}
and $\phi^{AA^{\prime}}$ and $\psi^{BB^{\prime}}$ are pure, bipartite
states~\cite{HDW08}. Their theorem applies to finite-dimensional systems, but
nevertheless, we apply their theorem to the infinite-dimensional setting by
means of a limiting argument.\footnote{The argument is similar to those
appearing Refs.~\cite{YS05,G08}, for example, and is simply that an
infinite-dimensional Hilbert space with a mean photon-number constraint is
effectively identical to a finite-dimensional Hilbert space. Suppose that we
truncate the Hilbert space at the channel input so that it is spanned by the
Fock number states $\left\{  \left\vert 0\right\rangle ,\left\vert
1\right\rangle ,\ldots,\left\vert K\right\rangle \right\}  $ where $K\gg
N_{S}$. Thus, all coherent states, squeezed states, and thermal states become
truncated to this finite-dimensional Hilbert space. Applying the
Hsieh-Devetak-Winter theorem to squeezed states in this truncated Hilbert
space gives a capacity region which is strictly an inner bound to the region
in (\ref{eq:EA-boson-region-1}-\ref{eq:EA-boson-region-3}). As we let $K$ grow
without bound, the entropies given by the Hsieh-Devetak-Winter theorem
converge to the entropies in (\ref{eq:EA-boson-region-1}%
-\ref{eq:EA-boson-region-3}).} By inspecting the above theorem, it becomes
clear that it is necessary to compute just seven entropies in order to
determine the achievable rate region: $H\left(  A\right)  _{\rho}$, $H\left(
B\right)  _{\rho}$, $H\left(  C\right)  _{\rho}$, $H\left(  AB\right)  _{\rho
}$, $H\left(  AC\right)  _{\rho}$, $H\left(  BC\right)  _{\rho}$, and
$H\left(  ABC\right)  _{\rho}$. Observe that $H\left(  ABC\right)  _{\rho
}=H\left(  E\right)  _{\rho}$ if we define $E$ as the environment of the
channel. In order to determine these entropies, we just need to figure out the
covariance matrices for each of the seven different systems corresponding to
these entropies because the entropies are a function of the symplectic
eigenvalues of these covariance matrices. These seven different covariance
matrices are as follows:%
\begin{gather}
V^{E}=%
\begin{bmatrix}
2\left(  \eta N_{S_{b}}+\left(  1-\eta\right)  N_{S_{a}}\right)  +1 & 0\\
0 & 2\left(  \eta N_{S_{b}}+\left(  1-\eta\right)  N_{S_{a}}\right)  +1
\end{bmatrix}
,\\
V^{A}=%
\begin{bmatrix}
2N_{S_{a}}+1 & 0\\
0 & 2N_{S_{a}}+1
\end{bmatrix}
,\\
V^{B}=%
\begin{bmatrix}
2N_{S_{b}}+1 & 0\\
0 & 2N_{S_{b}}+1
\end{bmatrix}
,\\
V^{C}=%
\begin{bmatrix}
2\left(  \eta N_{S_{a}}+\left(  1-\eta\right)  N_{S_{b}}\right)  +1 & 0\\
0 & 2\left(  \eta N_{S_{a}}+\left(  1-\eta\right)  N_{S_{b}}\right)  +1
\end{bmatrix}
,
\end{gather}%
\begin{equation}
V^{AB}=%
\begin{bmatrix}
2N_{S_{a}}+1 & 0 & 0 & 0\\
0 & 2N_{S_{a}}+1 & 0 & 0\\
0 & 0 & 2N_{S_{b}}+1 & 0\\
0 & 0 & 0 & 2N_{S_{b}}+1
\end{bmatrix}
,
\end{equation}%
\begin{equation}
V^{AC}=%
\begin{bmatrix}
2N_{S_{a}}+1 & 0 & 2\sqrt{\eta}\sqrt{N_{S_{a}}\left(  N_{S_{a}}+1\right)  } &
0\\
0 & 2N_{S_{a}}+1 & 0 & -2\sqrt{\eta}\sqrt{N_{S_{a}}\left(  N_{S_{a}}+1\right)
}\\
2\sqrt{\eta}\sqrt{N_{S_{a}}\left(  N_{S_{a}}+1\right)  } & 0 & 2\left(  \eta
N_{S_{a}}+\overline{\eta}N_{S_{b}}\right)  +1 & 0\\
0 & -2\sqrt{\eta}\sqrt{N_{S_{a}}\left(  N_{S_{a}}+1\right)  } & 0 & 2\left(
\eta N_{S_{a}}+\overline{\eta}N_{S_{b}}\right)  +1
\end{bmatrix}
, \label{eq:CM-AC}%
\end{equation}%
\begin{equation}
V^{BC}=%
\begin{bmatrix}
2\left(  \eta N_{S_{a}}+\overline{\eta}N_{S_{b}}\right)  +1 & 0 &
2\sqrt{\overline{\eta}}\sqrt{N_{S_{b}}\left(  N_{S_{b}}+1\right)  } & 0\\
0 & 2\left(  \eta N_{S_{a}}+\overline{\eta}N_{S_{b}}\right)  +1 & 0 &
-2\sqrt{\overline{\eta}}\sqrt{N_{S_{b}}\left(  N_{S_{b}}+1\right)  }\\
2\sqrt{\overline{\eta}}\sqrt{N_{S_{b}}\left(  N_{S_{b}}+1\right)  } & 0 &
2N_{S_{b}}+1 & 0\\
0 & -2\sqrt{\overline{\eta}}\sqrt{N_{S_{b}}\left(  N_{S_{b}}+1\right)  } & 0 &
2N_{S_{b}}+1
\end{bmatrix}
, \label{eq:CM-BC}%
\end{equation}
where $\overline{\eta}\equiv1-\eta$. The five entropies $H\left(  A\right)
_{\rho}$, $H\left(  B\right)  _{\rho}$, $H\left(  C\right)  _{\rho}$,
$H\left(  E\right)  _{\rho}$, and $H\left(  AB\right)  _{\rho}$ are
straightforward to compute because their covariance matrices all correspond to
those for thermal states:%
\begin{align}
H\left(  A\right)   &  =g\left(  N_{S_{a}}\right)  ,\\
H\left(  B\right)   &  =g\left(  N_{S_{b}}\right)  ,\\
H\left(  C\right)   &  =g\left(  \eta N_{S_{a}}+\left(  1-\eta\right)
N_{S_{b}}\right)  ,\\
H\left(  ABC\right)   &  =H\left(  E\right)  =g\left(  \eta N_{S_{b}}+\left(
1-\eta\right)  N_{S_{a}}\right)  ,\\
H\left(  AB\right)   &  =g\left(  N_{S_{a}}\right)  +g\left(  N_{S_{b}%
}\right)  .
\end{align}
We can calculate the other entropies $H\left(  AC\right)  $ and $H\left(
BC\right)  $ by computing the symplectic eigenvalues of the covariance
matrices in (\ref{eq:CM-AC}) and (\ref{eq:CM-BC}), respectively:%
\begin{align}
\lambda_{AC}^{\left(  \pm\right)  }  &  =\left(  1-\eta\right)  \left\vert
N_{S_{a}}-N_{S_{b}}\right\vert \pm\sqrt{\left(  1-\eta\right)  ^{2}\left(
N_{S_{a}}-N_{S_{b}}\right)  ^{2}+2\left(  1-\eta\right)  \left(  2N_{S_{a}%
}N_{S_{b}}+N_{S_{a}}+N_{S_{b}}\right)  +1},\\
\lambda_{BC}^{\left(  \pm\right)  }  &  =\eta\left\vert N_{S_{a}}-N_{S_{b}%
}\right\vert \pm\sqrt{\eta^{2}\left(  N_{S_{a}}-N_{S_{b}}\right)  ^{2}%
+2\eta\left(  2N_{S_{a}}N_{S_{b}}+N_{S_{a}}+N_{S_{b}}\right)  +1}.
\end{align}
Recall that we find the symplectic eigenvalues of a matrix $V$ by computing
the eigenvalues of the matrix $\left\vert iJV\right\vert $ \cite{WPGCRSL11}
where%
\begin{equation}
J\equiv\bigoplus_{i=1}^{n}%
\begin{bmatrix}
0 & 1\\
-1 & 0
\end{bmatrix}
,
\end{equation}
and $n$ is the number of modes. These symplectic eigenvalues lead to the
following values for the entropies:%
\begin{align}
H\left(  AC\right)   &  =g\left(  \left(  \lambda_{AC}^{+}+1\right)
/2-1\right)  +g\left(  \left(  \lambda_{AC}^{-}+1\right)  /2-1\right)  ,\\
H\left(  BC\right)   &  =g\left(  \left(  \lambda_{BC}^{+}+1\right)
/2-1\right)  +g\left(  \left(  \lambda_{BC}^{-}+1\right)  /2-1\right)  ,
\end{align}
by exploiting the fact that the entropy of a Gaussian state $\rho$ is the
following function of its symplectic eigenvalues $\left\{  \nu_{k}\right\}  $
\cite{WPGCRSL11}:%
\begin{equation}
H\left(  \rho\right)  =\sum_{k}g\left(  \left(  \nu_{k}+1\right)  /2-1\right)
.
\end{equation}
Thus, an achievable rate region for the entanglement-assisted bosonic multiple
access channel is as stated in the theorem.
\end{proof}

Figure~\ref{fig:EA-bosonic} plots several achievable rate regions given by
Theorem~\ref{thm:ea-bosonic} as the transmissivity parameter $\eta$ varies
from 0 to 1. The first plot has Alice's mean photon number much higher than
Bob's, while the second plot sets them equal.\begin{figure}[ptb]
\begin{center}
\includegraphics[
width=6.5337in
]{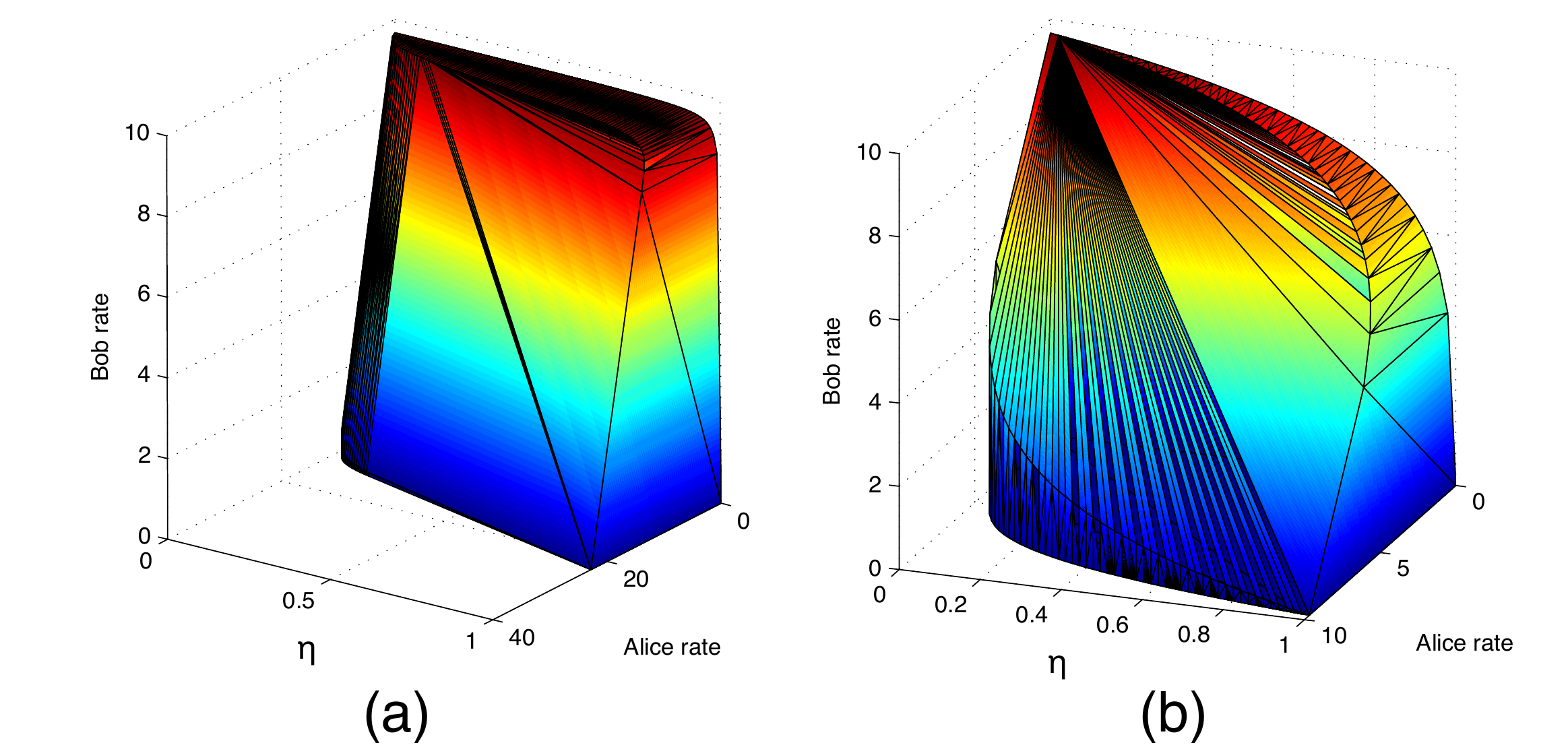}
\end{center}
\caption{The figure displays the achievable rate region from
Theorem~\ref{thm:ea-bosonic} for the beamsplitter multiple access channel as
the beamsplitter transmissivity $\eta$ varies from 0 to 1. (a) The region as
$\eta$ varies when Alice and Bob's mean input photon number are fixed at
$N_{S_{a}}=1000$ and $N_{S_{b}}=10$, respectively. (b) The region as $\eta$
varies when $N_{S_{a}}=10$ and $N_{S_{b}}=10$.}%
\label{fig:EA-bosonic}%
\end{figure}

\subsection{Comparison with the Unassisted Bosonic Multiple Access Rate
Region}

We would also like to compare the achievable rate region given by
Theorem~\ref{thm:ea-bosonic} to the Yen-Shapiro outer bound for unassisted
classical communication over the beamsplitter bosonic multiple access
channel~\cite{YS05}. Consider that the Yen-Shapiro outer bound is as follows:%
\begin{align}
R_{1}  &  \leq g\left(  N_{S_{a}}\right)  ,\\
R_{2}  &  \leq g\left(  N_{S_{b}}\right)  ,\\
R_{1}+R_{2}  &  \leq g\left(  \eta N_{S_{a}}+\left(  1-\eta\right)  N_{S_{b}%
}\right)  . \label{eq:unassisted-sum-rate}%
\end{align}
They derived this outer bound with two straightforward arguments. First, if
$N_{S_{a}}$ and $N_{S_{b}}$ are the respective mean photon numbers at the
channel input, then the mean photon number at the output is $\eta N_{S_{a}%
}+\left(  1-\eta\right)  N_{S_{b}}$, and the Holevo quantity can never exceed
$g\left(  \eta N_{S_{a}}+\left(  1-\eta\right)  N_{S_{b}}\right)
$~\cite{GGLMSY04}. The individual rate bounds follow by assuming that the
receiver gets access to both output ports of the channel. The best strategy
would then be simply to invert the beamsplitter, and the rate bounds follow
from a similar argument (that the Holevo quantity for mean photon number
constraints $N_{S_{a}}$ and $N_{S_{b}}$ cannot exceed $g\left(  N_{S_{a}%
}\right)  $ and $g\left(  N_{S_{b}}\right)  $, respectively).

It is straightforward to demonstrate that the sum rate bound in
Theorem~\ref{thm:ea-bosonic}\ always exceeds the sum rate bound in
(\ref{eq:unassisted-sum-rate}). Consider that the difference between these two
sum rate bounds is%
\begin{equation}
g\left(  N_{S_{a}}\right)  +g\left(  N_{S_{b}}\right)  -g\left(  \eta
N_{S_{b}}+\left(  1-\eta\right)  N_{S_{a}}\right)  ,
\end{equation}
and this quantity is always positive because $g\left(  x\right)  $ is positive
and monotone increasing for $x\geq0$ (i.e., supposing WLOG\ that $N_{S_{a}%
}\geq N_{S_{b}}$, it follows that $N_{S_{a}}\geq\eta N_{S_{b}}+\left(
1-\eta\right)  N_{S_{a}}$ and thus $g\left(  N_{S_{a}}\right)  \geq g\left(
\eta N_{S_{b}}+\left(  1-\eta\right)  N_{S_{a}}\right)  $). The individual
rate bounds are incomparable as Figure~\ref{fig:unassist-vs-assist}%
\ demonstrates---there are examples of channels and photon number constraints
for which the assisted region contains or does not contain the Yen-Shapiro
unassisted outer bound.

\begin{figure}[ptb]
\begin{center}
\includegraphics[
width=6.0in
]{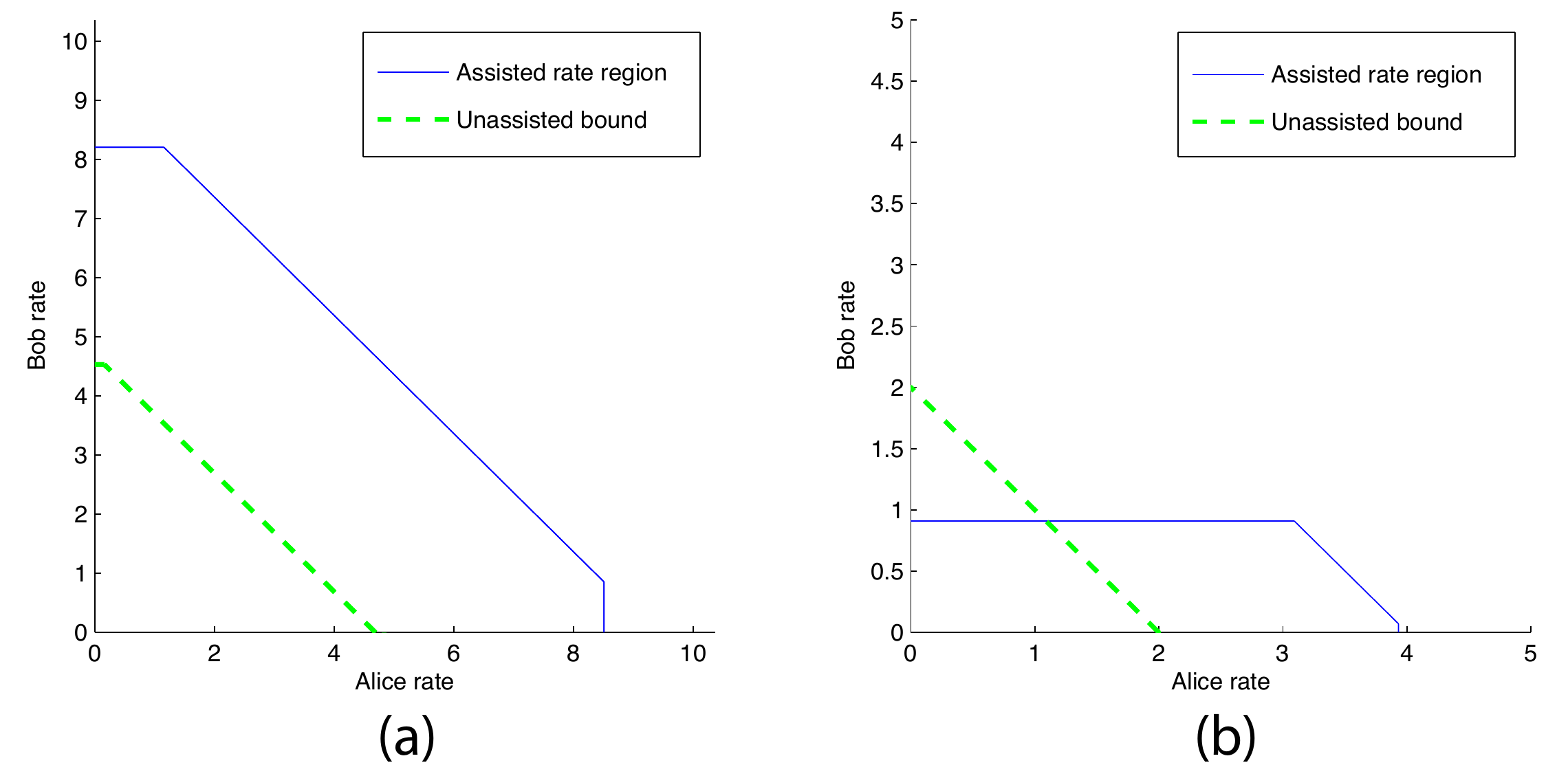}
\end{center}
\caption{The figure compares the achievable rate region from
Theorem~\ref{thm:ea-bosonic}~with the Yen-Shapiro outer bound on the
unassisted region~\cite{YS05} for two examples. (a) The first example shows
that our assisted achievable rate region contains the Yen-Shapiro outer bound
on the unassisted region when $N_{S_{a}}=10$, $N_{S_{b}}=8$, and $\eta=1/2$.
(b) The second example shows that our assisted achievable rate region does
\textit{not} contain the Yen-Shapiro outer bound on the unassisted region when
$N_{S_{a}}=1$, $N_{S_{b}}=1$, and $\eta=0.95$.}%
\label{fig:unassist-vs-assist}%
\end{figure}

\section{Conclusion}

We have discussed five different scenarios for entanglement-assisted classical
communication:\ sequential decoding for a single-sender, single-receiver
channel, sequential and successive decoding for a multiple access channel,
simultaneous decoding, coherent simultaneous decoding, and communication over
a bosonic channel. Our third contribution gives further progress toward
proving the quantum simultaneous decoding conjecture from Ref.~\cite{FHSSW11}%
\ (see Appendix~A in the thesis of Dutil for a different manifestation of this
conjecture in distributed compression~\cite{D11}).

Several open questions remain. It would of course be good to prove that the
quantum simultaneous decoding conjecture holds in the general case for
entanglement-assisted classical communication or even to broaden the classes
of channels or the conditions for which it holds. It would be worthwhile to
determine whether our strategy for entanglement-assisted classical
communication over a bosonic multiple access channel is optimal.

We are grateful to Vittorio Giovannetti for suggesting the idea of extending
the GLM\ sequential decoder to the entanglement-assisted case, and we thank
Pranab Sen for sharing his results in Ref.~\cite{S11a} and for pointing out
that a slight modification of our proof technique from the first version of
this article solves the quantum simultaneous decoding conjecture for two
senders. MMW\ acknowledges useful discussions with Omar Fawzi, Patrick Hayden,
Ivan Savov, and Pranab Sen during the development of Ref.~\cite{FHSSW11}.
MMW\ acknowledges financial support from the MDEIE\ (Qu\'{e}bec)
PSR-SIIRI\ international collaboration grant.

\appendix

\section{Appendix}

\label{sec:sequential-packing-proof}

\begin{proof}
[Proof of the Sequential Packing Lemma]Our proof below is essentially
identical to the proof given in Ref.~\cite{GLM10}, with the exception that it
extracts only the most basic conditions needed (these conditions are given in
the statement of the theorem). Given a message set $\mathcal{M}=\left\{
1,2,\dots,\left\vert \mathcal{M}\right\vert \right\}  $, we construct a code
$\mathcal{C}\equiv\left\{  c_{m}\right\}  _{m\in\mathcal{M}}$ randomly such
that each $c_{m}$ takes a value in $\mathcal{X}$ with probability
$p_{X}\left(  c_{m}\right)  $. Using this code, Alice chooses a message $m$
from the message set $\mathcal{M}$ and encodes it in the quantum codeword
$\rho_{c_{m}}$. To decode the message $m$, Bob performs the following steps:

\begin{enumerate}
\item Starting from $k=1$, Bob tries to determine if he received the $k$th message.

\item Bob first makes a projective measurement with the code subspace
projector $\Pi$ to determine if the received state is in the code subspace.

\item If the answer is NO, then an error has occurred and Bob aborts the protocol.

\item If the answer is YES, Bob performs another projective measurement on the
post-measurement state using the codeword subspace projector $\Pi_{c_{k}}$.

\item If the answer is YES, then Bob declares to have received the $k$th
message and stops the protocol.

\item If the answer is NO, then Bob increments $k$ and goes back to Step~2 if
$k<\left\vert \mathcal{M}\right\vert $. If $k=\left\vert \mathcal{M}%
\right\vert $, Bob declares that an error has occurred and aborts the protocol.
\end{enumerate}

As derived in Ref.~\cite{GLM10}, the following POVM $\left\{  \Lambda
_{m}\right\}  _{m\in\mathcal{M}}$ corresponds to the above sequential decoding
scheme:%
\begin{equation}
\Lambda_{m}\equiv\bar{Q}_{c_{1}}\cdots\bar{Q}_{c_{m-1}}\bar{\Pi}_{c_{m}}%
\bar{Q}_{c_{m-1}}\cdots\bar{Q}_{c_{1}},
\end{equation}
where for any operator $\Theta$, we define $\bar{\Theta}$ as%
\begin{equation}
\bar{\Theta}\equiv\Pi\Theta\Pi,
\end{equation}
and%
\begin{equation}
Q_{x}\equiv I-\Pi_{x}.
\end{equation}

We analyze the performance of this sequential decoding scheme by computing a
lower bound on the expectation of the average success probability, where the
expectation is with respect to all possible codes:%
\begin{align}
\mathbb{E}_{\mathcal{C}}\left\{  \bar{p}_{\text{succ}}\left(  \mathcal{C}%
\right)  \right\}   &  =\sum_{c_{1},\dots,c_{\left\vert \mathcal{M}\right\vert
}}p_{X}\left(  c_{1}\right)  \cdots p_{X}\left(  c_{\left\vert \mathcal{M}%
\right\vert }\right)  \frac{1}{\left\vert \mathcal{M}\right\vert }\sum
_{m=1}^{\left\vert \mathcal{M}\right\vert }\mathrm{Tr}\left\{  \Pi_{c_{m}}%
\bar{Q}_{c_{m-1}}\cdots\bar{Q}_{c_{1}}\rho_{c_{m}}\bar{Q}_{c_{1}}\cdots\bar
{Q}_{c_{m-1}}\right\} \\
&  =\frac{1}{\left\vert \mathcal{M}\right\vert }\sum_{l=0}^{\left\vert
\mathcal{M}\right\vert -1}\sum_{x,c_{1},\dots,c_{l}}p_{X}\left(  x\right)
p_{X}\left(  c_{1}\right)  \cdots p_{X}\left(  c_{l}\right)  \mathrm{Tr}%
\left\{  \Pi_{x}\bar{Q}_{c_{l}}\cdots\bar{Q}_{c_{1}}\rho_{x}\bar{Q}_{c_{1}%
}\cdots\bar{Q}_{c_{l}}\right\} \\
&  =\frac{1}{\left\vert \mathcal{M}\right\vert }\sum_{l=0}^{\left\vert
\mathcal{M}\right\vert -1}\sum_{x,c_{1},\dots,c_{l}}p_{X}\left(  x\right)
p_{X}\left(  c_{1}\right)  \cdots p_{X}\left(  c_{l}\right)  \sum_{y}%
\sum_{y^{\prime}\in\mathcal{T}_{x}}\lambda_{x,y}\left\vert \left\langle
\psi_{x,y^{\prime}}\right\vert \bar{Q}_{c_{1}}\cdots\bar{Q}_{c_{l}}\left\vert
\psi_{x,y}\right\rangle \right\vert ^{2}. \label{eq:succ-lower-bound-1}%
\end{align}
We obtain the last equality by writing out the spectral decomposition of
$\Pi_{x}$ and $\rho_{x}$:%
\begin{align}
\rho_{x}  &  =\sum_{y}\lambda_{x,y}\left\vert \psi_{x,y}\right\rangle
\left\langle \psi_{x,y}\right\vert ,\\
\Pi_{x}  &  =\sum_{y\in\mathcal{T}_{x}}\left\vert \psi_{x,y}\right\rangle
\left\langle \psi_{x,y}\right\vert .
\end{align}
We note that $\rho_{x}$ and $\Pi_{x}$ commute by assumption and therefore
share common eigenstates. We use $\mathcal{T}_{x}$ to index a subset of the
eigenstates of $\rho_{x}$.

The following lower bound applies to the rightmost term in
(\ref{eq:succ-lower-bound-1}):%
\begin{align}
\sum_{y}\sum_{y^{\prime}\in\mathcal{T}_{x}}\lambda_{x,y}\left\vert
\left\langle \psi_{x,y^{\prime}}\right\vert \bar{Q}_{c_{1}}\cdots\bar
{Q}_{c_{l}}\left\vert \psi_{x,y}\right\rangle \right\vert ^{2}  &  \geq
\sum_{y\in\mathcal{T}_{x}}\lambda_{x,y}\left\vert \left\langle \psi
_{x,y}\right\vert \bar{Q}_{c_{1}}\cdots\bar{Q}_{c_{l}}\left\vert \psi
_{x,y}\right\rangle \right\vert ^{2}\\
&  =\sum_{y\in\mathcal{T}_{x}}\lambda_{x,y}\left\vert \left\langle \psi
_{x,y}\right\vert \bar{Q}_{c_{1}}\cdots\bar{Q}_{c_{l}}\left\vert \psi
_{x,y}\right\rangle \right\vert ^{2}\sum_{y}\lambda_{x,y}\label{eq:beforeCS}\\
&  \geq\left\vert \sum_{y\in\mathcal{T}_{x}}\lambda_{x,y}\left\langle
\psi_{x,y}\right\vert \bar{Q}_{c_{1}}\cdots\bar{Q}_{c_{l}}\left\vert
\psi_{x,y}\right\rangle \right\vert ^{2}\label{eq:afterCS}\\
&  =\left\vert \mathrm{Tr}\left\{  \Pi_{x}\rho_{x}\Pi_{x}\bar{Q}_{c_{1}}%
\cdots\bar{Q}_{c_{l}}\right\}  \right\vert ^{2}.
\end{align}
The first inequality follows by eliminating some positive terms from the
summation. The second inequality follows by applying the Cauchy-Schwarz
inequality. The last equality follows by exploiting the assumed commutative
relation between $\Pi_{x}$ and $\rho_{x}$. Therefore, the following lower
bound applies to the expectation of the average success probability:%
\begin{equation}
\mathbb{E}_{\mathcal{C}}\left\{  \bar{p}_{\text{succ}}\left(  \mathcal{C}%
\right)  \right\}  \geq\frac{1}{\left\vert \mathcal{M}\right\vert }\sum
_{l=0}^{\left\vert \mathcal{M}\right\vert -1}\sum_{x,c_{1},\dots,c_{l}}%
p_{X}\left(  x\right)  p_{X}\left(  c_{1}\right)  \cdots p_{X}\left(
c_{l}\right)  \left\vert \mathrm{Tr}\left\{  \Pi_{x}\rho_{x}\Pi_{x}\bar
{Q}_{c_{1}}\cdots\bar{Q}_{c_{l}}\right\}  \right\vert ^{2}.
\end{equation}
We again apply the Cauchy-Schwarz inequality to the inner summation:%
\begin{align}
&  \sum_{x,c_{1},\dots,c_{l}}p_{X}\left(  x\right)  p_{X}\left(  c_{1}\right)
\cdots p_{X}\left(  c_{l}\right)  \left\vert \mathrm{Tr}\left\{  \Pi_{x}%
\rho_{x}\Pi_{x}\bar{Q}_{c_{1}}\cdots\bar{Q}_{c_{l}}\right\}  \right\vert
^{2}\\
\geq &  \left\vert \sum_{x,c_{1},\dots,c_{l}}p_{X}\left(  x\right)
p_{X}\left(  c_{1}\right)  \cdots p_{X}\left(  c_{l}\right)  \mathrm{Tr}%
\left\{  \Pi_{x}\rho_{x}\Pi_{x}\bar{Q}_{c_{1}}\cdots\bar{Q}_{c_{l}}\right\}
\right\vert ^{2}\\
=  &  \left\vert \mathrm{Tr}\left\{  \left(  \sum_{x}p_{X}\left(  x\right)
\Pi_{x}\rho_{x}\Pi_{x}\right)  \left(  \sum_{c_{1}}p_{X}\left(  c_{1}\right)
\bar{Q}_{c_{1}}\right)  \cdots\left(  \sum_{c_{l}}p_{X}\left(  c_{l}\right)
\bar{Q}_{c_{l}}\right)  \right\}  \right\vert ^{2}\\
=  &  \left\vert \mathrm{Tr}\left\{  W_{1}\mathcal{Q}^{l}\right\}  \right\vert
^{2},
\end{align}
where we define%
\begin{align}
W_{q}  &  \equiv\sum_{x}p_{X}\left(  x\right)  \Pi_{x}\rho_{x}^{q}\Pi_{x},\\
\mathcal{Q}  &  \equiv\sum_{x}p_{X}\left(  x\right)  \bar{Q}_{x},
\end{align}
and it is understood that $\mathcal{Q}^{0}=\Pi$ (an abuse of notation
explained further below). Therefore, we obtain the following lower bound on
the expectation of the average success probability:%
\begin{equation}
\mathbb{E}_{\mathcal{C}}\left\{  \bar{p}_{\text{succ}}\left(  \mathcal{C}%
\right)  \right\}  \geq\frac{1}{\left\vert \mathcal{M}\right\vert }\sum
_{l=0}^{\left\vert \mathcal{M}\right\vert -1}\left\vert \mathrm{Tr}\left\{
W_{1}\mathcal{Q}^{l}\right\}  \right\vert ^{2}. \label{eq:succ-prob-1-1}%
\end{equation}

In order to proceed, we note that
\begin{align}
\mathcal{Q}  &  =\sum_{x}p_{X}\left(  x\right)  \bar{Q}_{x}\\
&  =\Pi\left(  \sum_{x}p_{X}\left(  x\right)  \left(  I-\Pi_{x}\right)
\right)  \Pi\\
&  \leq I,
\end{align}
and therefore%
\begin{equation}
\mathrm{Tr}\left\{  W_{1}\mathcal{Q}^{l}\right\}  =\mathrm{Tr}\left\{
W_{1}\mathcal{Q}^{\frac{l-1}{2}}\mathcal{QQ}^{\frac{l-1}{2}}\right\}
\leq\mathrm{Tr}\left\{  W_{1}\mathcal{Q}^{l-1}\right\}  .
\end{equation}

Given this observation, we can further lower bound the success probability by
taking the smallest term of the summation from (\ref{eq:succ-prob-1-1}):%
\begin{align}
\mathbb{E}_{\mathcal{C}}\left\{  \bar{p}_{\text{succ}}\left(  \mathcal{C}%
\right)  \right\}   &  \geq\left\vert \mathrm{Tr}\left\{  W_{1}\mathcal{Q}%
^{\left\vert \mathcal{M}\right\vert -1}\right\}  \right\vert ^{2}\\
&  =\left\vert \mathrm{Tr}\left\{  W_{1}\left(  \bar{I}-\bar{W}_{0}\right)
^{\left\vert \mathcal{M}\right\vert -1}\right\}  \right\vert ^{2}\\
&  =\left\vert \mathrm{Tr}\left\{  \sum_{z=0}^{\left\vert \mathcal{M}%
\right\vert -1}\binom{\left\vert \mathcal{M}\right\vert -1}{z}\left(
-1\right)  ^{z}W_{1}\bar{I}^{\left\vert \mathcal{M}\right\vert -z}\bar{W}%
_{0}^{z}\right\}  \right\vert ^{2}\\
&  =\left\vert \sum_{z=0}^{\left\vert \mathcal{M}\right\vert -1}%
\binom{\left\vert \mathcal{M}\right\vert -1}{z}\left(  -1\right)
^{z}\mathrm{Tr}\left\{  W_{1}\Pi\bar{W}_{0}^{z}\right\}  \right\vert ^{2}%
\end{align}
Here, we define a function $f_{z}$ as%
\begin{equation}
f_{z}\equiv\mathrm{Tr}\left\{  W_{1}\Pi\bar{W}_{0}^{z}\right\}  ,
\end{equation}
where $z$ is a nonnegative integer. We abused the notation of $\bar{W}_{0}%
^{z}$ here, which does \textit{not }mean to raise the eigenvalues of $\bar
{W}_{0}$ to the power $z$ in its spectral decomposition, but rather%
\begin{equation}
\bar{W}_{0}^{z}=\prod_{i=1}^{z}\bar{W_{0}},
\end{equation}
as it arises from the binomial expansion. We note that $f_{0}=\mathrm{Tr}%
\left\{  W_{1}\Pi\right\}  $ and the function $f_{z}$ is always positive.
Thus, the above expression is equal to the following one:%
\begin{align}
&  \left\vert \sum_{z=0}^{\left\vert \mathcal{M}\right\vert -1}\binom
{\left\vert \mathcal{M}\right\vert -1}{z}\left(  -1\right)  ^{z}%
f_{z}\right\vert ^{2}\\
&  =\left\vert f_{0}+\sum_{z=1}^{\left\vert \mathcal{M}\right\vert -1}%
\binom{\left\vert \mathcal{M}\right\vert -1}{z}\left(  -1\right)  ^{z}%
f_{z}\right\vert ^{2}\\
&  =\left\vert A\right\vert ^{2},
\end{align}
with%
\begin{equation}
A\equiv f_{0}+\sum_{z=1}^{\left\vert \mathcal{M}\right\vert -1}\binom
{\left\vert \mathcal{M}\right\vert -1}{z}\left(  -1\right)  ^{z}f_{z}.
\end{equation}
We then have%
\begin{equation}
A\geq2f_{0}-\sum_{z=0}^{\left\vert \mathcal{M}\right\vert -1}\binom{\left\vert
\mathcal{M}\right\vert -1}{z}f_{z}.
\end{equation}
The function $f_{z}$ satisfies the following two properties:%
\begin{align}
f_{0}  &  \geq1-2\epsilon,\\
f_{z}  &  \leq\left(  \frac{d}{D}\right)  ^{z}f_{0}.
\end{align}
We now prove this. First we show that $f_{0}$ is $\epsilon$-close to one:
\allowdisplaybreaks[1]%
\begin{align}
f_{0}  &  =\mathrm{Tr}\left\{  W_{1}\Pi\right\} \\
&  =\sum_{x}p_{X}\left(  x\right)  \mathrm{Tr}\left\{  \Pi_{x}\rho_{x}\Pi
_{x}\Pi\right\} \\
&  =\sum_{x}p_{X}\left(  x\right)  \mathrm{Tr}\left\{  \Pi_{x}\rho_{x}%
\Pi\right\} \\
&  =\sum_{x}p_{X}\left(  x\right)  \mathrm{Tr}\left\{  \left(  I-\left(
I-\Pi_{x}\right)  \right)  \rho_{x}\Pi\right\} \\
&  =\sum_{x}p_{X}\left(  x\right)  \mathrm{Tr}\left\{  \rho_{x}\Pi\right\}
-\sum_{x}p_{X}\left(  x\right)  \mathrm{Tr}\left\{  (I-\Pi_{x})\rho_{x}%
\Pi\right\} \\
&  \geq\sum_{x}p_{X}\left(  x\right)  \mathrm{Tr}\left\{  \rho_{x}\Pi\right\}
-\sum_{x}p_{X}\left(  x\right)  \mathrm{Tr}\left\{  (I-\Pi_{x})\rho
_{x}\right\} \\
&  \geq\sum_{x}p_{X}\left(  x\right)  \mathrm{Tr}\left\{  \rho_{x}\Pi\right\}
-\epsilon\nonumber\\
&  \geq1-2\epsilon
\end{align}
The third equality follows by the commutative relation between $\Pi_{x}$ and
$\rho_{x}$. The second inequality follows from the condition
(\ref{eq:unit-prob-2}). Now we will upper bound the function $f_{z}$ in terms
of $f_{z-1}$, and we show that we can upper bound $f_{z}$ in terms of
$f_{z-1}$, and as a result, in terms of $f_{0}$.
\begin{align}
f_{z}  &  =\mathrm{Tr}\left\{  W_{1}\bar{W}_{0}^{z}\right\} \\
&  =\mathrm{Tr}\left\{  \sqrt{W_{1}}\bar{W}_{0}^{\frac{z-1}{2}}\bar{W}_{0}%
\bar{W}_{0}^{\frac{z-1}{2}}\sqrt{W_{1}}\right\} \\
&  =\mathrm{Tr}\left\{  \sqrt{W_{1}}\bar{W}_{0}^{\frac{z-1}{2}}\Pi\left(
\sum_{x}p_{X}\left(  x\right)  \Pi_{x}\right)  \Pi\bar{W}_{0}^{\frac{z-1}{2}%
}\sqrt{W_{1}}\right\} \\
&  \leq d\cdot\mathrm{Tr}\left\{  \sqrt{W_{1}}\bar{W}_{0}^{\frac{z-1}{2}}%
\Pi\left(  \sum_{x}p_{X}\left(  x\right)  \Pi_{x}\rho_{x}\Pi_{x}\right)
\Pi\bar{W}_{0}^{\frac{z-1}{2}}\sqrt{W_{1}}\right\} \\
&  \leq d\cdot\mathrm{Tr}\left\{  \sqrt{W_{1}}\bar{W}_{0}^{\frac{z-1}{2}}%
\Pi\rho\Pi\bar{W}_{0}^{\frac{z-1}{2}}\sqrt{W_{1}}\right\} \\
&  \leq\frac{d}{D}\mathrm{Tr}\left\{  \sqrt{W_{1}}\bar{W}_{0}^{\frac{z-1}{2}%
}\Pi\bar{W}_{0}^{\frac{z-1}{2}}\sqrt{W_{1}}\right\} \\
&  \leq\frac{d}{D}\mathrm{Tr}\left\{  W_{1}\bar{W}_{0}^{z-1}\right\} \\
&  =\frac{d}{D}f_{z-1}\\
\Rightarrow f_{z}  &  \geq\left(  \frac{d}{D}\right)  ^{z}f_{0}%
\end{align}
In this derivation, we used the conditions (\ref{eq:equi-part-1}) and
(\ref{eq:equi-part-2}) and the fact that $W_{1}$ and $\bar{W}_{0}$ are positive.

Therefore, using the above two inequalities, we get that%
\begin{align}
A  &  \geq2f_{0}-f_{0}\sum_{z=0}^{\left\vert \mathcal{M}\right\vert -1}%
\binom{\left\vert \mathcal{M}\right\vert -1}{z}\left(  \frac{d}{D}\right)
^{z}\\
&  =f_{0}\left(  2-\left(  1+\frac{d}{D}\right)  ^{\left\vert \mathcal{M}%
\right\vert -1}\right) \\
&  \geq\left(  1-2\epsilon\right)  \left(  2-e^{\frac{d}{D}\left\vert
\mathcal{M}\right\vert }\right)  .
\end{align}
The last inequality follows from the fact that $1+x\leq e^{x}$ for all $x$,
and our analysis completes with the observation that%
\begin{equation}
\mathbb{E}_{\mathcal{C}}\left\{  \bar{p}_{\text{succ}}\left(  \mathcal{C}%
\right)  \right\}  \geq\left\vert A\right\vert ^{2}\geq\left\vert \left(
1-2\epsilon\right)  \left(  2-e^{\frac{d}{D}\left\vert \mathcal{M}\right\vert
}\right)  \right\vert ^{2},
\end{equation}
as long as $2-\exp\left\{  d\left\vert \mathcal{M}\right\vert /D\right\}  $ is positive.
\end{proof}

\bibliographystyle{plain}
\bibliography{Ref}

\begin{thebibliography}{10}

\bibitem{PhysRevA.56.3470}
Christoph Adami and Nicolas~J. Cerf.
\newblock von {Neumann} capacity of noisy quantum channels.
\newblock {\em Physical Review A}, 56(5):3470--3483, November 1997.

\bibitem{PhysRevLett.70.1895}
Charles~H. Bennett, Gilles Brassard, Claude Cr\'epeau, Richard Jozsa, Asher
  Peres, and William~K. Wootters.
\newblock Teleporting an unknown quantum state via dual classical and
  {Einstein-Podolsky-Rosen} channels.
\newblock {\em Physical Review Letters}, 70(13):1895--1899, March 1993.

\bibitem{BSST99}
Charles~H. Bennett, Peter~W. Shor, John~A. Smolin, and Ashish~V. Thapliyal.
\newblock Entanglement-assisted classical capacity of noisy quantum channels.
\newblock {\em Physical Review Letters}, 83(15):3081--3084, October 1999.

\bibitem{BSST02}
Charles~H. Bennett, Peter~W. Shor, John~A. Smolin, and Ashish~V. Thapliyal.
\newblock Entanglement-assisted capacity of a quantum channel and the reverse
  {Shannon} theorem.
\newblock {\em IEEE Transactions on Information Theory}, 48:2637--2655, 2002.

\bibitem{PhysRevLett.69.2881}
Charles~H. Bennett and Stephen~J. Wiesner.
\newblock Communication via one- and two-particle operators on
  {Einstein-Podolsky-Rosen} states.
\newblock {\em Physical Review Letters}, 69(20):2881--2884, November 1992.

\bibitem{B04}
Garry Bowen.
\newblock Quantum feedback channels.
\newblock {\em IEEE Transactions in Information Theory}, 50(10):2429--2434,
  October 2004.
\newblock arXiv:quant-ph/0209076.

\bibitem{book1991cover}
Thomas~M. Cover and Joy~A. Thomas.
\newblock {\em Elements of Information Theory}.
\newblock Wiley-Interscience, 1991.

\bibitem{CK11}
Imre Csisz\'{a}r and J\'{a}nos K\"{o}rner.
\newblock {\em Information Theory: Coding Theorems for Discrete Memoryless
  Systems}.
\newblock Cambridge University Press, second edition, August 2011.

\bibitem{DR09}
Nilanjana Datta and Renato Renner.
\newblock Smooth entropies and the quantum information spectrum.
\newblock {\em IEEE Transactions on Information Theory}, 55(6):2807--2815, June
  2009.
\newblock arXiv:0801.0282.

\bibitem{DHW05RI}
Igor Devetak, Aram~W. Harrow, and Andreas Winter.
\newblock A resource framework for quantum {Shannon} theory.
\newblock {\em IEEE Transactions on Information Theory}, 54(10):4587--4618,
  October 2008.

\bibitem{D11}
Nicolas Dutil.
\newblock {\em Multiparty quantum protocols for assisted entanglement
  distillation}.
\newblock PhD thesis, McGill University, May 2011.
\newblock arXiv:1105.4657.

\bibitem{EW07}
Jens Eisert and Michael~M. Wolf.
\newblock {\em Quantum Information with Continous Variables of Atoms and
  Light}, chapter Gaussian quantum channels, pages 23--42.
\newblock Imperial College Press, London, 2007.
\newblock arXiv:quant-ph/0505151.

\bibitem{el2010lecture}
Abbas El~Gamal and Young-Han Kim.
\newblock {Lecture notes on network information theory}.
\newblock January 2010.
\newblock arXiv:1001.3404.

\bibitem{FHSSW11}
Omar Fawzi, Patrick Hayden, Ivan Savov, Pranab Sen, and Mark~M. Wilde.
\newblock Classical communication over a quantum interference channel.
\newblock February 2011.
\newblock arXiv:1102.2624.

\bibitem{F54}
Amiel Feinstein.
\newblock A new basic theorem of information theory.
\newblock {\em IEEE Transactions on Information Theory}, 4(4):2--22, 1954.

\bibitem{GK04}
Christopher Gerry and Peter Knight.
\newblock {\em Introductory Quantum Optics}.
\newblock Cambridge University Press, November 2004.

\bibitem{GGLMSY04}
Vittorio Giovannetti, Saikat Guha, Seth Lloyd, Lorenzo Maccone, Jeffrey~H.
  Shapiro, and Horace~P. Yuen.
\newblock Classical capacity of the lossy bosonic channel: The exact solution.
\newblock {\em Physical Review Letters}, 92(2):027902, January 2004.

\bibitem{GLM10}
Vittorio Giovannetti, Seth Lloyd, and Lorenzo Maccone.
\newblock Achieving the {Holevo} bound via sequential measurements.
\newblock December 2010.
\newblock arXiv:1012.0386.

\bibitem{GLMS03a}
Vittorio Giovannetti, Seth Lloyd, Lorenzo Maccone, and Peter~W. Shor.
\newblock Broadband channel capacities.
\newblock {\em Physical Review A}, 68(6):062323, December 2003.

\bibitem{GLMS03}
Vittorio Giovannetti, Seth Lloyd, Lorenzo Maccone, and Peter~W. Shor.
\newblock Entanglement assisted capacity of the broadband lossy channel.
\newblock {\em Physical Review Letters}, 91(4):047901, July 2003.

\bibitem{G08}
Saikat Guha.
\newblock {\em Multiple-User Quantum Information Theory for Optical
  Communication Channels}.
\newblock PhD thesis, Massachusetts Institute of Technology, June 2008.

\bibitem{HK81}
Te~Sun Han and Kingo Kobayashi.
\newblock A new achievable rate region for the interference channel.
\newblock {\em IEEE Transactions on Information Theory}, 27(1):49--60, January
  1981.

\bibitem{Har03}
A.~W. Harrow.
\newblock Coherent communication of classical messages.
\newblock {\em Physical Review Letters}, 92:097902, 2004.

\bibitem{hayashi2003general}
Masahito Hayashi and Hiroshi Nagaoka.
\newblock General formulas for capacity of classical-quantum channels.
\newblock {\em IEEE Transactions on Information Theory}, 49(7):1753--1768,
  2003.

\bibitem{H02}
Alexander~S. Holevo.
\newblock On entanglement assisted classical capacity.
\newblock {\em Journal of Mathematical Physics}, 43(9):4326--4333, 2002.

\bibitem{HW01}
Alexander~S. Holevo and Reinhard~F. Werner.
\newblock Evaluating capacities of bosonic {Gaussian} channels.
\newblock {\em Physical Review A}, 63(3):032312, February 2001.

\bibitem{nature2005horodecki}
M.~Horodecki, Jonathan Oppenheim, and Andreas Winter.
\newblock Partial quantum information.
\newblock {\em Nature}, 436:673--676, 2005.

\bibitem{HDW08}
Min-Hsiu Hsieh, Igor Devetak, and Andreas Winter.
\newblock Entanglement-assisted capacity of quantum multiple-access channels.
\newblock {\em IEEE Transactions on Information Theory}, 54(7):3078--3090,
  2008.

\bibitem{HW08GFP}
Min-Hsiu Hsieh and Mark~M. Wilde.
\newblock Entanglement-assisted communication of classical and quantum
  information.
\newblock {\em IEEE Transactions on Information Theory}, 56(9):4682--4704,
  September 2010.

\bibitem{HW09}
Min-Hsiu Hsieh and Mark~M. Wilde.
\newblock Trading classical communication, quantum communication, and
  entanglement in quantum {Shannon} theory.
\newblock {\em IEEE Transactions on Information Theory}, 56(9):4705--4730,
  September 2010.

\bibitem{book2000mikeandike}
Michael~A. Nielsen and Isaac~L. Chuang.
\newblock {\em Quantum Computation and Quantum Information}.
\newblock Cambridge University Press, 2000.

\bibitem{ON07}
Tomohiro Ogawa and Hiroshi Nagaoka.
\newblock Making good codes for classical-quantum channel coding via quantum
  hypothesis testing.
\newblock {\em IEEE Transactions on Information Theory}, 53(6):2261--2266, June
  2007.

\bibitem{S11}
Pranab Sen.
\newblock Private communication, July 2011.

\bibitem{S11a}
Pranab Sen.
\newblock Sequential decoding for some channels with classical input and
  quantum output, July 2011.

\bibitem{bell1948shannon}
Claude~E. Shannon.
\newblock A mathematical theory of communication.
\newblock {\em Bell System Technical Journal}, 27:379--423, 1948.

\bibitem{WPGCRSL11}
Christian Weedbrook, Stefano Pirandola, Raul Garcia-Patron, Nicolas~J. Cerf,
  Timothy~C. Ralph, Jeffrey~H. Shapiro, and Seth Lloyd.
\newblock Gaussian quantum information.
\newblock 2011.
\newblock In preparation.

\bibitem{W11}
Mark~M. Wilde.
\newblock {\em From Classical to Quantum {Shannon} Theory}.
\newblock June 2011.
\newblock arXiv:1106.1445.

\bibitem{itit1999winter}
Andreas Winter.
\newblock Coding theorem and strong converse for quantum channels.
\newblock {\em IEEE Transactions on Information Theory}, 45(7):2481--2485,
  1999.

\bibitem{W01}
Andreas Winter.
\newblock The capacity of the quantum multiple access channel.
\newblock {\em IEEE Transactions on Information Theory}, 47:3059--3065, 2001.

\bibitem{YHD05MQAC}
Jon Yard, Patrick Hayden, and Igor Devetak.
\newblock Capacity theorems for quantum multiple-access channels:
  Classical-quantum and quantum-quantum capacity regions.
\newblock {\em IEEE Transactions on Information Theory}, 54(7):3091--3113,
  2008.

\bibitem{YS05}
Brent~J. Yen and Jeffrey~H. Shapiro.
\newblock Multiple-access bosonic communications.
\newblock {\em Physical Review A}, 72(6):062312, December 2005.

\end{thebibliography}

\end{document}